\documentclass[letterpaper, 10 pt, conference]{ieeeconf}  

\IEEEoverridecommandlockouts                              
\usepackage[T1]{fontenc}
\let\labelindent\relax 

\usepackage{dsfont}
\usepackage{amsfonts}
\usepackage{amsmath}
\usepackage{amssymb}
\usepackage{mathtools}
\usepackage{pifont}
\usepackage{bm}  
\usepackage{bbm}
\usepackage{mathrsfs}
\usepackage[shortlabels]{enumitem}

\DeclareMathOperator*{\argmin}{arg\,min}
\usepackage{cite}
\usepackage{booktabs}
\usepackage{multicol}
\usepackage{multirow}
\usepackage{tabularx}
\usepackage{diagbox}
\usepackage{graphicx}
\usepackage{subcaption}
\usepackage{array}
\usepackage{hyperref}
 \hypersetup{
     colorlinks=true,
     linkcolor=blue,
     filecolor=blue,
     citecolor = black,      
     urlcolor=cyan,
     }
\usepackage{cleveref}
\usepackage[dvipsnames]{xcolor}
\usepackage{tikz}
\usepackage{verbatim}
\usepackage{listings}
\definecolor{mygreen}{RGB}{28,172,0} 
\definecolor{mylilas}{RGB}{170,55,241}

\usepackage{matlab-prettifier}
\usepackage{color}

\usepackage{tikz,pgfplots}
\usepackage{pgfplots}
\usepackage{pgfplotstable}
\definecolor{gray2}{HTML}{ededed}
\definecolor{gray3}{HTML}{F5F5F5}
\definecolor{RoyalAzure}{rgb}{0.0, 0.22, 0.66}
\usetikzlibrary{shapes.geometric,backgrounds,patterns, trees}
\usetikzlibrary{3d,decorations.text,shapes.arrows,positioning,fit,backgrounds}
\usetikzlibrary{positioning, decorations.pathmorphing, shapes}
\usetikzlibrary{decorations.pathreplacing}
\usetikzlibrary{shapes.geometric,backgrounds,patterns, trees}
\usetikzlibrary{spy}
\usetikzlibrary{arrows.meta,
                bending,
                intersections,
                quotes,
                shapes.geometric}
              \usetikzlibrary{automata, positioning}
              \usepgfplotslibrary{fillbetween}
\usetikzlibrary{shapes,arrows}
\usetikzlibrary{arrows.meta}
\usetikzlibrary{positioning}
\tikzset{set/.style={draw,circle,inner sep=0pt,align=center}}
\usetikzlibrary{automata, positioning}
  \usetikzlibrary{shapes,shadows}
  \tikzstyle{abstractbox} = [draw=black, fill=white, rectangle,
  inner sep=10pt, style=rounded corners, drop shadow={fill=black,
  opacity=1}]
\tikzstyle{abstracttitle} =[fill=white]
\usetikzlibrary{calc,positioning,shapes.geometric}
\usetikzlibrary{arrows.meta,arrows}

\usepackage[titlenumbered,ruled,linesnumbered]{algorithm2e}
\usepackage[noend]{algpseudocode}
\makeatletter
\def\BState{\State\hskip-\ALG@thistlm}
\makeatother
\newcommand{\K}{\mathrm{k}}
\newcommand{\D}{\mathrm{D}}
\newcommand{\A}{\mathrm{A}}
\newcommand{\supp}{\operatorname{supp}}

\newcommand{\E}{\mathbb{E}}
\newcommand{\pr}{\mathbb{P}}

\newcommand{\N}{\mathbb{N}}
\newcommand{\R}{\mathbb{R}}

\newtheorem{theorem}{Theorem}
\newtheorem{lemma}{Lemma}
\newtheorem{corollary}{Corollary}

\newtheorem{definition}{Definition}

\newtheorem{assumption}{Assumption}

\usepackage{tikz,pgfplots}
\usepackage{pgfplots}
\usepackage{pgfplotstable}
\definecolor{gray2}{HTML}{ededed}
\definecolor{gray3}{HTML}{F5F5F5}
\definecolor{RoyalAzure}{rgb}{0.0, 0.22, 0.66}
\usetikzlibrary{shapes.geometric,backgrounds,patterns, trees}
\usetikzlibrary{3d,decorations.text,shapes.arrows,positioning,fit,backgrounds}
\usetikzlibrary{positioning, decorations.pathmorphing, shapes}
\usetikzlibrary{decorations.pathreplacing}
\usetikzlibrary{shapes.geometric,backgrounds,patterns, trees}
\usetikzlibrary{spy}
\usetikzlibrary{arrows.meta,
                bending,
                intersections,
                quotes,
                shapes.geometric}
              \usetikzlibrary{automata, positioning}
              \usepgfplotslibrary{fillbetween}
\usetikzlibrary{shapes,arrows}
\usetikzlibrary{arrows.meta}
\usetikzlibrary{positioning}
\tikzset{set/.style={draw,circle,inner sep=0pt,align=center}}
\usetikzlibrary{automata, positioning}
  \usetikzlibrary{shapes,shadows}
  \tikzstyle{abstractbox} = [draw=black, fill=white, rectangle,
  inner sep=10pt, style=rounded corners, drop shadow={fill=black,
  opacity=1}]
\tikzstyle{abstracttitle} =[fill=white]
\usetikzlibrary{calc,positioning,shapes.geometric}
\usetikzlibrary{arrows.meta,arrows}

\makeatletter
\pgfkeys{/pgf/.cd,
  parallelepiped offset x/.initial=2mm,
  parallelepiped offset y/.initial=2mm
}
\pgfdeclareshape{parallelepiped}
{
  \inheritsavedanchors[from=rectangle] 
  \inheritanchorborder[from=rectangle]
  \inheritanchor[from=rectangle]{north}
  \inheritanchor[from=rectangle]{north west}
  \inheritanchor[from=rectangle]{north east}
  \inheritanchor[from=rectangle]{center}
  \inheritanchor[from=rectangle]{west}
  \inheritanchor[from=rectangle]{east}
  \inheritanchor[from=rectangle]{mid}
  \inheritanchor[from=rectangle]{mid west}
  \inheritanchor[from=rectangle]{mid east}
  \inheritanchor[from=rectangle]{base}
  \inheritanchor[from=rectangle]{base west}
  \inheritanchor[from=rectangle]{base east}
  \inheritanchor[from=rectangle]{south}
  \inheritanchor[from=rectangle]{south west}
  \inheritanchor[from=rectangle]{south east}
  \backgroundpath{
    \southwest \pgf@xa=\pgf@x \pgf@ya=\pgf@y
    \northeast \pgf@xb=\pgf@x \pgf@yb=\pgf@y
    \pgfmathsetlength\pgfutil@tempdima{\pgfkeysvalueof{/pgf/parallelepiped
      offset x}}
    \pgfmathsetlength\pgfutil@tempdimb{\pgfkeysvalueof{/pgf/parallelepiped
      offset y}}
    \def\ppd@offset{\pgfpoint{\pgfutil@tempdima}{\pgfutil@tempdimb}}
    \pgfpathmoveto{\pgfqpoint{\pgf@xa}{\pgf@ya}}
    \pgfpathlineto{\pgfqpoint{\pgf@xb}{\pgf@ya}}
    \pgfpathlineto{\pgfqpoint{\pgf@xb}{\pgf@yb}}
    \pgfpathlineto{\pgfqpoint{\pgf@xa}{\pgf@yb}}
    \pgfpathclose
    \pgfpathmoveto{\pgfqpoint{\pgf@xb}{\pgf@ya}}
    \pgfpathlineto{\pgfpointadd{\pgfpoint{\pgf@xb}{\pgf@ya}}{\ppd@offset}}
    \pgfpathlineto{\pgfpointadd{\pgfpoint{\pgf@xb}{\pgf@yb}}{\ppd@offset}}
    \pgfpathlineto{\pgfpointadd{\pgfpoint{\pgf@xa}{\pgf@yb}}{\ppd@offset}}
    \pgfpathlineto{\pgfqpoint{\pgf@xa}{\pgf@yb}}
    \pgfpathmoveto{\pgfqpoint{\pgf@xb}{\pgf@yb}}
    \pgfpathlineto{\pgfpointadd{\pgfpoint{\pgf@xb}{\pgf@yb}}{\ppd@offset}}
  }
}

\makeatletter
\tikzset{anchor/.append code=\let\tikz@auto@anchor\relax,
  add font/.code=%
    \expandafter\def\expandafter\tikz@textfont\expandafter{\tikz@textfont#1},
  left delimiter/.style 2 args={append after command={\tikz@delimiter{south east}
    {south west}{every delimiter,every left delimiter,#2}{south}{north}{#1}{.}{\pgf@y}}}}
\tikzstyle{sms} = [rectangle callout, draw,very thick, rounded corners, minimum height=20pt]
\makeatletter
\tikzset{anchor/.append code=\let\tikz@auto@anchor\relax,
  add font/.code=%
    \expandafter\def\expandafter\tikz@textfont\expandafter{\tikz@textfont#1},
  left delimiter/.style 2 args={append after command={\tikz@delimiter{south east}
    {south west}{every delimiter,every left delimiter,#2}{south}{north}{#1}{.}{\pgf@y}}}}
\tikzstyle{sms} = [rectangle callout, draw,very thick, rounded corners, minimum height=20pt]
\usetikzlibrary{positioning,calc}
\tikzstyle{block} = [rectangle, draw,
text width=10.5em, text centered, rounded corners, minimum height=4em]
\tikzstyle{line} = [draw, -latex]
\tikzset{
  mybackground9/.style={execute at end picture={
        \begin{scope}[on background layer]
          \draw[black,fill=black!5,rounded corners=6ex] (current bounding box.south west)
                    rectangle (current bounding box.north east);
          \node[draw,fill=white,ellipse,anchor=west,inner sep=1pt,minimum width=4ex] at (current bounding box.north
                   west){#1};
        \end{scope}
    }},
}

\tikzset{
  mybackground13/.style={execute at end picture={
        \begin{scope}[on background layer]
          \draw[black, fill=gray2, rounded corners=4ex] (current bounding box.south west)
                    rectangle (current bounding box.north east);
          \node[draw,fill=white,ellipse,anchor=west,inner sep=1pt,minimum width=4ex] at (current bounding box.north
                   west){#1};
        \end{scope}
    }},
}
\tikzset{
  mybackground14/.style={execute at end picture={
        \begin{scope}[on background layer]
          \draw[black, rounded corners=2ex] (current bounding box.south west)
                    rectangle (current bounding box.north east);
          \node[draw,fill=white,ellipse,anchor=west,inner sep=1pt,minimum width=4ex] at (current bounding box.north
                   west){#1};
        \end{scope}
    }},
}

\tikzset{
  mybackground6/.style={execute at end picture={
        \begin{scope}[on background layer]
          \draw[black,rounded corners=1ex, line width=0.15mm] (current bounding box.south west)
                    rectangle (current bounding box.north east);
          \node[draw,fill=white,ellipse,anchor=west,inner sep=1pt,minimum width=4ex] at (current bounding box.north
                   west){#1};
        \end{scope}
    }},
}

\tikzset{
  mybackground11/.style={execute at end picture={
        \begin{scope}[on background layer]
          \draw[black, fill=Black!80!Sepia!9, rounded corners=6ex] (current bounding box.south west)
                    rectangle (current bounding box.north east);
          \node[draw,fill=white,ellipse,anchor=west,inner sep=1pt,minimum width=4ex] at (current bounding box.north
                   west){#1};
        \end{scope}
    }},
}

\tikzset{
  mybackground15/.style={execute at end picture={
        \begin{scope}[on background layer]
          \draw[black, fill=Black!80!Sepia!9, rounded corners=3ex] (current bounding box.south west)
                    rectangle (current bounding box.north east);
          \node[draw,fill=white,ellipse,anchor=west,inner sep=1pt,minimum width=4ex] at (current bounding box.north
                   west){#1};
        \end{scope}
    }},
}

\tikzset{
  mybackground12/.style={execute at end picture={
        \begin{scope}[on background layer]
          \draw[black, fill=Black!40!Emerald!30, rounded corners=3ex, line width=0.3mm] (current bounding box.south west)
                    rectangle (current bounding box.north east);
        \end{scope}
    }},
}
\tikzset{
  mybackground18/.style={execute at end picture={
      \begin{scope}[on background layer]
        \draw[black, fill=gray3, rounded corners=3.5ex] (current bounding box.south west)
        rectangle (current bounding box.north east);
        \node[draw,fill=white,ellipse,anchor=west,inner sep=1pt,minimum width=4ex] at (current bounding box.north
        west){#1};
      \end{scope}
    }}
}
\tikzset{
  mybackground58/.style={execute at end picture={
        \begin{scope}[on background layer]
          \draw[black, fill=blue!40!black!5, rounded corners=1ex] (current bounding box.south west)
                    rectangle (current bounding box.north east);
          \node[draw,fill=white,ellipse,anchor=west,inner sep=1pt,minimum width=4ex, rounded corners=1ex] at (current bounding box.north
                   west){#1};
        \end{scope}
    }},
}
\tikzset{l3 switch/.style={
    parallelepiped,fill=switch, draw=white,
    minimum width=0.75cm,
    minimum height=0.75cm,
    parallelepiped offset x=1.75mm,
    parallelepiped offset y=1.25mm,
    path picture={
      \node[fill=white,
        circle,
        minimum size=6pt,
        inner sep=0pt,
        append after command={
          \pgfextra{
            \foreach \angle in {0,45,...,360}
            \draw[-latex,fill=white] (\tikzlastnode.\angle)--++(\angle:2.25mm);
          }
        }
      ]
       at ([xshift=-0.75mm,yshift=-0.5mm]path picture bounding box.center){};
    }
  },
  ports/.style={
    line width=0.3pt,
    top color=gray!20,
    bottom color=gray!80
  },
  rack switch/.style={
    parallelepiped,fill=white, draw,
    minimum width=1.25cm,
    minimum height=0.25cm,
    parallelepiped offset x=2mm,
    parallelepiped offset y=1.25mm,
    xscale=-1,
    path picture={
      \draw[top color=gray!5,bottom color=gray!40]
      (path picture bounding box.south west) rectangle
      (path picture bounding box.north east);
      \coordinate (A-west) at ([xshift=-0.2cm]path picture bounding box.west);
      \coordinate (A-center) at ($(path picture bounding box.center)!0!(path
        picture bounding box.south)$);
      \foreach \x in {0.275,0.525,0.775}{
        \draw[ports]([yshift=-0.05cm]$(A-west)!\x!(A-center)$)
          rectangle +(0.1,0.05);
        \draw[ports]([yshift=-0.125cm]$(A-west)!\x!(A-center)$)
          rectangle +(0.1,0.05);
       }
      \coordinate (A-east) at (path picture bounding box.east);
      \foreach \x in {0.085,0.21,0.335,0.455,0.635,0.755,0.875,1}{
        \draw[ports]([yshift=-0.1125cm]$(A-east)!\x!(A-center)$)
          rectangle +(0.05,0.1);
      }
    }
  },
  server/.style={
    parallelepiped,
    fill=white, draw,
    minimum width=0.35cm,
    minimum height=0.75cm,
    parallelepiped offset x=3mm,
    parallelepiped offset y=2mm,
    xscale=-1,
    path picture={
      \draw[top color=gray!5,bottom color=gray!40]
      (path picture bounding box.south west) rectangle
      (path picture bounding box.north east);
      \coordinate (A-center) at ($(path picture bounding box.center)!0!(path
        picture bounding box.south)$);
      \coordinate (A-west) at ([xshift=-0.575cm]path picture bounding box.west);
      \draw[ports]([yshift=0.1cm]$(A-west)!0!(A-center)$)
        rectangle +(0.2,0.065);
      \draw[ports]([yshift=0.01cm]$(A-west)!0.085!(A-center)$)
        rectangle +(0.15,0.05);
      \fill[black]([yshift=-0.35cm]$(A-west)!-0.1!(A-center)$)
        rectangle +(0.235,0.0175);
      \fill[black]([yshift=-0.385cm]$(A-west)!-0.1!(A-center)$)
        rectangle +(0.235,0.0175);
      \fill[black]([yshift=-0.42cm]$(A-west)!-0.1!(A-center)$)
        rectangle +(0.235,0.0175);
    }
  },
}

\usetikzlibrary{calc, shadings, shadows, shapes.arrows}
\tikzset{cross/.style={cross out, draw=black, minimum size=2*(#1-\pgflinewidth), inner sep=0pt, outer sep=0pt},
cross/.default={1pt}}
\tikzset{%
  interface/.style={draw, rectangle, rounded corners, font=\LARGE\sffamily},
  ethernet/.style={interface, fill=yellow!50},
  serial/.style={interface, fill=green!70},
  speed/.style={sloped, anchor=south, font=\large\sffamily},
  route/.style={draw, shape=single arrow, single arrow head extend=4mm,
    minimum height=1.7cm, minimum width=3mm, white, fill=switch!20,
    drop shadow={opacity=.8, fill=switch}, font=\tiny}
}
\makeatletter
\pgfdeclareradialshading[tikz@ball]{cloud}{\pgfpoint{-0.275cm}{0.4cm}}{%
  color(0cm)=(tikz@ball!75!white);
  color(0.1cm)=(tikz@ball!85!white);
  color(0.2cm)=(tikz@ball!95!white);
  color(0.7cm)=(tikz@ball!89!black);
  color(1cm)=(tikz@ball!75!black)
}
\tikzoption{cloud color}{\pgfutil@colorlet{tikz@ball}{#1}%
  \def\tikz@shading{cloud}\tikz@addmode{\tikz@mode@shadetrue}}
\makeatother

\tikzset{my cloud/.style={
     cloud, draw, aspect=2,
     cloud color={gray!5!white}
  }
}

\usetikzlibrary{backgrounds}
\usetikzlibrary{patterns}
\pgfmathdeclarefunction{gauss}{3}{%
  \pgfmathparse{1/(#3*sqrt(2*pi))*exp(-((#1-#2)^2)/(2*#3^2))}%
}
\pgfmathdeclarefunction{cdf}{3}{%
  \pgfmathparse{1/(1+exp(-0.07056*((#1-#2)/#3)^3 - 1.5976*(#1-#2)/#3))}%
}
\pgfmathdeclarefunction{fq}{3}{%
  \pgfmathparse{1/(sqrt(2*pi*#1))*exp(-(sqrt(#1)-#2/#3)^2/2)}%
}
\pgfmathdeclarefunction{fq0}{1}{%
  \pgfmathparse{1/(sqrt(2*pi*#1))*exp(-#1/2)}%
}
\renewcommand\thetable{\arabic{table}}
\usepackage{etoolbox}

\begin{document}
\title{\LARGE \bf Conjectural Online Learning with First-order Beliefs \\in Asymmetric Information Stochastic Games
}



\author{ Tao Li, Kim Hammar, Rolf Stadler, and Quanyan Zhu
\thanks{Correspondence should be addressed to T. Li. The online appendix is available at \url{https://arxiv.org/pdf/2402.18781}. This work is partially supported by grants ECCS-1847056.
}
\thanks{T. Li and Q. Zhu are with the Department of Electrical and Computer Engineering, New York University, USA \texttt{\{tl2636, qz494\}@nyu.edu}}
\thanks{
K. Hammar and R. Stadler are with the Division of Network and Systems Engineering, KTH Royal Institute of Technology, Sweden, \texttt{\{kimham, stadler\}@kth.se}
}
}

\maketitle

\begin{abstract}
Asymmetric information stochastic games (\textsc{aisg}s) arise in many complex socio-technical systems, such as cyber-physical systems and IT infrastructures. Existing computational methods for \textsc{aisg}s are primarily offline and can not adapt to equilibrium deviations. Further, current methods are limited to particular information structures to avoid belief hierarchies. {Considering these limitations}, we propose conjectural online learning (\textsc{col}), an online learning method under generic information structures in \textsc{aisg}s. \textsc{col} uses a forecaster-actor-critic (\textsc{fac}) architecture, where subjective forecasts are used to conjecture the opponents' strategies within a lookahead horizon, and Bayesian learning is used to calibrate the conjectures. To adapt strategies to nonstationary environments based on information feedback, \textsc{col} uses online rollout with cost function approximation (actor-critic). We prove that the conjectures produced by \textsc{col} are asymptotically consistent with the information feedback in the sense of a relaxed Bayesian consistency. We also prove that the empirical strategy profile induced by \textsc{col} converges to the Berk-Nash equilibrium, a solution concept characterizing rationality under subjectivity. Experimental results from an intrusion response use case demonstrate \textsc{col}'s {faster convergence} over state-of-the-art reinforcement learning methods against nonstationary attacks.
\end{abstract}
\section{Introduction}
Stochastic game theory provides an analytical framework for automated and resilient management of complex socio-technical systems (\textsc{sts}s) \cite{tao22confluence}, such as cyber-physical systems and IT infrastructures \cite{li2023decision}, where decision-making entities (players) jointly control the system's evolution. Due to the complex nature of \textsc{sts}s and players' distinct capabilities, players have \textit{asymmetric} information structures (feedback), i.e., players acquire different information over time.

Information asymmetry poses a significant computational challenge since each player has to reason recursively about the other players' (opponents') private information. In particular, information asymmetry causes players' beliefs to bifurcate, which means that each player has to form a belief about the opponents' belief (a second-order belief), a belief about the opponents' second-order belief (a third-order belief), etc., leading to infinite belief hierarchies \cite{solan_game}.

To avoid belief hierarchies, prior work focuses on special classes of asymmetric information stochastic games (\textsc{aisg}s) where all players share the same belief, e.g., one-sided partially observable games \cite{horak23posg,kim_posg}, stochastic games with public observations \cite{horak19public-obs,kim_gamesec23}, hidden stochastic games \cite{ziliotto20hiddenSG}, and common-information-based equilibria \cite{gupta13CIB-MPE, ouyang16dyna_asy}. While the assumption that players share the same belief simplifies computations, it is not realistic for most practical scenarios. Moreover, prior work focuses on offline methods for equilibrium computation, which means that the obtained strategies become irrelevant if opponents deviate from the equilibrium path during online execution.

This paper presents an attempt at a unified online learning framework under generic information structures. We study the following question: \textit{how should a player reason about the opponent's private information and update its strategy online in generic \textsc{aisg}s?} To answer this question, we present conjectural online learning (\textsc{col}), an online method where each player a) uses first-order beliefs that admit simple Bayesian updates; and b) conjectures that the opponent's strategy is selected from a candidate set. The belief and the conjecture constitute the player's subjective perception of the game and the learning process, which is calibrated based on information feedback, see \Cref{fig:pipeline}.
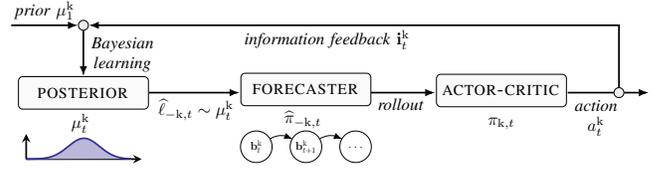
\begin{figure}
  \centering
  \scalebox{0.8}{
    \begin{tikzpicture}

\node[scale=1] (kth_cr) at (5.9,0.46)
{
\begin{tikzpicture}

  \def\B{11};
  \def\Bs{3.0};
  \def\xmax{\B+3.2*\Bs};
  \def\ymin{{-0.1*gauss(\B,\B,\Bs)}};
  \def\h{0.08*gauss(\B,\B,\Bs)};
  \def\N{50}

  \begin{axis}[every axis plot post/.append style={
      mark=none,domain=0:20,
      samples=\N,smooth},
               xmin=0, xmax=20,
               ymin=0, ymax={1.1*gauss(\B,\B,\Bs)},
               axis lines=middle,
               axis line style=thick,
               enlargelimits=upper, 
               ticks=none,
               every axis x label/.style={at={(current axis.right of origin)},anchor=north},
               width=3.5cm,
               height=2cm,
               clip=false
              ]

    \addplot[Blue,thick,name path=B] {gauss(x,\B,\Bs)};
    \path[name path=xaxis](0,0) -- (20,0);
    \addplot[Blue!25] fill between[of=xaxis and B];
  \end{axis}
  \node[inner sep=0pt,align=center, scale=0.8, color=black] (hacker) at (0.95,0.6) {
    $\mu^{\K}_t$
  };

\node[scale=0.9] (box) at (0.95,1.1)
{
\begin{tikzpicture}
  \draw[rounded corners=0.5ex, fill=black!2] (0,0) rectangle node (m1){} (2.4,0.6);
  \node[inner sep=0pt,align=center, scale=1, color=black] (hacker) at (1.25,0.3) {
    \textsc{posterior}
  };
\end{tikzpicture}
};
\end{tikzpicture}
};

\node[scale=1] (kth_cr) at (9.6,0.42)
{
\begin{tikzpicture}

\node[scale=0.75] (box) at (2.75,1.35)
{
\begin{tikzpicture}
\node[draw,circle, minimum width=14mm, scale=0.5](s0) at (0.8,1.3) {\Large$\mathbf{b}^{\K}_{t}$};
\node[draw,circle, minimum width=10mm, scale=0.5](s1) at (1.9,1.3) {\Large$\mathbf{b}^{\K}_{t+1}$};
\node[draw,circle, minimum width=12.7mm, scale=0.55](s3) at (3,1.3) {\Large$\hdots$};

\draw[-{Latex[length=1.7mm]}, bend left] (s0) to (s1);
\draw[-{Latex[length=1.7mm]}, bend left] (s1) to (s3);
\end{tikzpicture}
};

  \node[inner sep=0pt,align=center, scale=0.8, color=black] (hacker) at (2.75,1.8) {
    $\widehat{\pi}_{-\K,t}$
  };

\node[scale=0.9] (box) at (2.75,2.3)
{
\begin{tikzpicture}
  \draw[rounded corners=0.5ex, fill=black!2] (0,0) rectangle node (m1){} (2.4,0.6);
  \node[inner sep=0pt,align=center, scale=1, color=black] (hacker) at (1.25,0.3) {
    \textsc{forecaster}
  };
\end{tikzpicture}
};
\end{tikzpicture}
};

\node[scale=1] (kth_cr) at (12.85,0.73)
{
\begin{tikzpicture}

  \node[inner sep=0pt,align=center, scale=0.8, color=black] (hacker) at (2.8,1.95) {
    $\pi_{\mathrm{k},t}$
  };
\node[scale=0.9] (box) at (2.75,2.5)
{
\begin{tikzpicture}
  \draw[rounded corners=0.5ex, fill=black!2] (0,0) rectangle node (m1){} (2.4,0.6);
  \node[inner sep=0pt,align=center, scale=1, color=black] (hacker) at (1.25,0.3) {
    \textsc{actor-critic}
  };
\end{tikzpicture}
};
\end{tikzpicture}
};

%
%

  \node[inner sep=0pt,align=center, scale=0.8, color=black] (hacker) at (7.82,0.55) {
    $\widehat{\ell}_{-\K, t} \sim \mu_t^{\mathrm{k}}$
  };
  \node[inner sep=0pt,align=center, scale=0.8, color=black] (hacker) at (11.22,0.63) {
    \textit{rollout}
  };
  \node[inner sep=0pt,align=center, scale=0.8, color=black] (hacker) at (14.42,0.415) {
    \textit{action}\\
    $a_t^{\K}$
  };
  \node[inner sep=0pt,align=center, scale=0.8, color=black] (hacker) at (5.31,2.2) {
    \textit{prior} $\mu_1^{\K}$
  };
  \node[inner sep=0pt,align=center, scale=0.8, color=black] (hacker) at (6.55,1.46) {
    \textit{Bayesian}\\
    \textit{learning}
  };
  \node[inner sep=0pt,align=center, scale=0.8, color=black] (hacker) at (10,1.71) {
    \textit{information feedback} $\mathbf{i}^{\K}_t$
  };
  \node[draw,circle, fill=gray2, scale=0.5] (feedback) at (14.8,0.83) {};
  \node[draw,circle, fill=gray2, scale=0.5](prior) at (5.9,1.95) {};
  \draw[-{Latex[length=1.7mm]}, thick] (6.99, 0.83) to (8.5, 0.83);
  \draw[-{Latex[length=1.7mm]}, thick] (10.68, 0.83) to (11.75, 0.83);
  \draw[-, thick] (13.95, 0.83) to (feedback);
  \draw[-{Latex[length=1.7mm]}, thick]  (feedback) to (15.3, 0.83);
  \draw[-{Latex[length=1.7mm]}, thick]  (feedback) to (14.8, 1.95) to (prior);
  \draw[-{Latex[length=1.7mm]}, thick] (4.7, 1.95) to (prior);
  \draw[-{Latex[length=1.7mm]}, thick] (prior) to (5.9, 1.1);
\end{tikzpicture}
  }
  \caption{One-step cycle in \textsc{col}: conjectural online learning (see also Alg. \ref{alg:online_rollout}); the player $\mathrm{k}$ updates its conjecture $\widehat{\ell}_{-\K,t}$ about the opponent's policy parameterization by sampling from the posterior $\mu_{t}^{\K}$, from which it forecasts opponent's future moves $\widehat{\pi}_{-\K,t}$ conditional on its own first-order beliefs $\mathbf{b}_t^{\K}$; a rollout-based actor-critic creates policy improvement against the conjectured opponent.}
  \label{fig:pipeline}
  \vspace{-0.4cm}
\end{figure}

\textsc{col} is based on a forecaster-actor-critic (\textsc{fac}) architecture where the forecaster conjectures the opponent's future strategies by selecting one from the candidate set at each time step. The conjecture is then used to estimate the expected cost (critic) and to update the strategy through rollout (actor). The conjecture is subsequently updated through Bayesian learning based on information feedback.

To quantify the accuracy of a conjecture, we define a notion of \textit{conjecture consistency} based on the \textsc{kl} divergence between the subjective conjecture and the objective information feedback. This consistency notion allows us to characterize the asymptotic convergence of Bayesian learning (see Thm. \ref{thm:asym-consistency}) and to prove that the empirical strategy profile induced by \textsc{col} converges to the Berk-Nash equilibrium \cite{esponda16berk} (see Cor. \ref{coro:convergence-berk-nash}). \textbf{Our contributions} are summarized below.
\begin{enumerate}
    \item We introduce conjectural online learning (\textsc{col}), an online method for \textsc{aisg}s where each player iteratively adapts its conjecture using Bayesian learning and updates its strategy through rollout. \textsc{col} allows a player to adapt its strategy to a nonstationary opponent.
    \item We prove that \textsc{col} asymptotically converges to a Berk-Nash equilibrium in \textsc{aisg}s, where the limiting conjecture is consistent with the information feedback.
    \item We evaluate \textsc{col} on an intrusion response use case based on measurements from a testbed running 64 virtualized servers and 10 different types of intrusions, showing that \textsc{col} {adapts faster to nonstationary attacks than} current reinforcement learning methods.
\end{enumerate}

\section{Asymmetric Information Stochastic Game}
\label{sec:aisg}
Consider an infinite-horizon, discrete-time (indexed by $t\in \mathbb{N}$), finite, stochastic game $\Gamma$ with asymmetric information where players do not perfectly observe the states or actions:
\begin{align}
\Gamma\triangleq\langle \mathcal{N}, \mathcal{S}, \{\mathcal{A}^\K, \mathcal{O}^\K, z^\K, c^\K\}_{\K \in \mathcal{N}}, f, \mathbf{b}_1 , \gamma \rangle. \label{eq:game_def}
\end{align}
$\mathcal{N}$ is the set of players, indexed by $\K\in \mathcal{N}$. $\mathcal{S}$ is the set of states, unobservable to all players. $\mathcal{A}^\K, \mathcal{O}^\K$ are the sets of actions and observations, respectively. All the sets mentioned above are assumed to be finite and endowed with the discrete topology. $z^\K: \mathcal{S} \rightarrow \Delta(\mathcal{O}^\K)$ is the observation kernel, where $\Delta(\cdot)$ is the set of Borel probability measures over the underlying set.  $c^\K: \mathcal{S}\times \prod_{\K\in \mathcal{N}}\mathcal{A}^\K\rightarrow \R $ is the cost function, $f: \mathcal{S}\times  \prod_{\K\in \mathcal{N}}\mathcal{A}^\K\rightarrow \Delta(\mathcal{S})$ is the transition function, and $\mathbf{b}_1\in \Delta(\mathcal{S})$ is the initial state distribution. $\gamma\in [0,1)$ is a discounting factor. 

In addition to the above, some helpful notations are as follows. Elements of the aforementioned sets are denoted by the corresponding lowercase letters. Boldface lowercase letters (e.g., $\mathbf{x}$) denote vectors. A random variable is written in upper case (e.g., $X$), a random vector in boldface (e.g., $\mathbf{X}$). $|\mathcal{S}|$ denotes the cardinality of the set $\mathcal{S}$.

The game begins with a randomly sampled initial state $s_1\sim \mathbf{b}_1$. At each time step, each player observes a private partial observation $o^\K_t\sim z(\cdot|s_t)$. We assume perfect recall, i.e., players perfectly recollect the play history. Denote by $\mathbf{i}_t^{\K}\triangleq\{a_{t-1}^{\K}, o^\K_t\}$ the information feedback revealed to player $\K$ before its decision-making at time $t$ \cite{tao_info}. The history can then be recursively defined as $\mathbf{h}^{\K}_t\triangleq (\mathbf{h}^{\K}_{l-1}, \mathbf{i}_l^{\K} )_{l=2,\hdots,t} \in \mathcal{H}^{\mathrm{k}}$, with $\mathbf{h}^{\K}_1\triangleq\{\mathbf{b}_1, o^\K_1\}$. After observing the partial observation and updating the history, the player selects an action $a_t^\K$ according to its strategy $\pi_\K: \mathcal{H}^\K\rightarrow\Delta(\mathcal{A}^\K)$, which incurs a stage cost $c_t^\K(s_t, a_t^{\K}, a_t^{-\K})$ determined by the joint actions. Finally, the game transitions to a new state $s_{t+1}$, after which the above procedure is repeated.

Each player aims to minimize its expected cumulative cost, defined below
\begin{equation}
\label{eq:objective}
    J_\K^{(\pi_\K, \pi_{-\K})}(\mathbf{b}_1)\triangleq \E\left[\sum_{t=1}^{\infty}\gamma^{t-1}c^{\K}(S_t, A^{\K}_t, A^{-\K}_t)\mid \mathbf{b}_1 \right].
\end{equation}

A strategy $\pi_\K$ is a \textit{best response} against $\pi_{-\K}$ if it minimizes $J_\K^{(\pi_\K, \pi_{-\K})}$. Such a best response may not be unique in general, and hence, we use $\operatorname{BR}_\K(\pi_{-\K})\triangleq \argmin_{\pi_\K}J_\K^{(\pi_\K, \pi_{-\K})}$ to represent the best response correspondence. If each player follows the best response strategy against the opponents' strategies, then no player has an incentive to deviate from their strategy, and the resulting strategy profile is referred to as a Nash equilibrium, defined as
\begin{equation}
\label{eq:nash}
   \bm{\pi}^\star= (\pi_\K^\star, \pi_{-\K}^\star)\in \operatorname{BR}_\K(\pi_{-\K}^\star)\times \operatorname{BR}_{-\K}(\pi^\star_{\K}).
\end{equation}

\subsection{Asymmetric Information and Belief Hierarchy}
\label{subsec:nested-beliefs}
Player $\K$'s belief about the hidden state at time $t$ is the conditional probability of the underlying state given the realized history: $\mathbf{b}^\K_t\triangleq \pr[s_t|\mathbf{h}_t^\K]$. The belief state $\mathbf{b}_t^\K$ is a sufficient statistic for $\mathbf{h}_t^\K$ \cite{kumar82stochastic}. Consequently, the strategy can be defined on the belief space $\mathcal{B}\triangleq \Delta(S)$, i.e., $\pi_\K:\mathcal{B}\rightarrow \Delta(\mathcal{A}^\K)$. For simplicity, we consider the two-player case and refer to player $-\K$ as the opponent in the sequel.

Suppose that the opponent is aware of player $\K$'s belief state and uses the same belief to determine its actions as $\pi_{-\K}(\cdot| \mathbf{b}_t^\K)$, then each player can compute $\mathbf{b}_t^\K$ recursively through \eqref{eq:belief-update}. While this assumption makes the computation of $\mathbf{b}_t^\K$ tractable, it is unrealistic since, in practice, each player has their own private information, leading to separate belief states for each player. That is, even though players share the initial belief state $\mathbf{b}_1^{\K}=\mathbf{b}_1$, their belief evolutions bifurcate as the information feedbacks $\mathbf{i}_t^\K$ and $\mathbf{i}_t^{-\K}$ differ. For this reason, each player also has a belief over the opponents' beliefs, leading to second-order beliefs from $\Delta(\mathcal{S}\times \mathcal{B})$. The second-order beliefs also bifurcate, leading to third-order beliefs and so on, creating infinite hierarchies of beliefs.
\begin{figure*}[!t]
\small
\begin{equation}
\label{eq:belief-update}
\mathscr{B}(\mathbf{b}^\K_{t-1}, \mathbf{i}_t^\K, \pi_{-\K})(s_t=s) \triangleq \frac{z^\K(o^\K_{t} \mid s_t=s)\sum_{\tilde{s} \in \mathcal{S}}\sum_{a \in \mathcal{A}^{-\K}}\pi_{-\K}(a^{-\K}_{t-1}=a \mid \mathbf{b}^\K_{t-1})\mathbf{b}^\K_{t-1}(s_{t-1}=\tilde{s})f(s \mid \tilde{s}, a^{\K}_{t-1},a^{-\K}_{t-1})}{\sum_{a \in \mathcal{A}^{-\K}}\sum_{s^{\prime},\tilde{s}\in \mathcal{S}}z^\K(o^\K_{t}\mid s_t=s^{\prime})\pi_{-\K}(a^{-\K}_{t-1}=a \mid  \mathbf{b}^\K_{t-1})\mathbf{b}^{\K}_{t-1}(s_{t-1}=\tilde{s})f(s^{\prime} \mid \tilde{s}, a^{\K}_{t-1},a^{-\K}_{t-1})}.
\end{equation}
\vspace{-0.7cm}
\end{figure*}

Even though prior works bear distinct motivations and theoretical treatments, most of them focus on subclasses of \textsc{aisg} where a common belief is shared among all players, allowing for the avoidance of belief hierarchies. For example, \cite{horak23posg,kim_posg} considers one-sided partial observability, where one informed player, say $\K$, can observe the state, the other players' observations, as well as the other player's actions, i.e., it has full observability. The information structures in this type of game are $\mathbf{i}_t^\K=\{a_{t-1}^\K, a_{t-1}^{-\K}, o_{t}^{-\K}, s_t\}$ and $\mathbf{i}_t^{-\K}=\{a_{t-1}^{-\K}, o_t^{-\K}\}$. The informed player reconstructs the belief of the lesser informed, and thus both players share the same belief. Following a similar approach, \cite{horak19public-obs,ziliotto20hiddenSG} assume public observations that are shared across players, i.e., $\mathbf{i}_t^\K=\{o_t^\K, a_{t-1}^{\K}\}$, $o_t^\K=o_t^{-\K}=o_t$. In this type of game, the public observations enable each player to reconstruct the opponent's belief. Consequently, players reach a consensus over the joint belief state $(\mathbf{b}_t^\K, \mathbf{b}_t^{-\K})$. In a separate line of work, \cite{gupta13CIB-MPE,ouyang16dyna_asy} explore common-information-based beliefs. The common information refers to the non-empty intersection of information structures $\mathbf{i}_t^{c}\triangleq \cap_{\K\in \mathcal{N}}\mathbf{i}_t^{\K}$, which subsequently leads to common beliefs among players given that only $\mathbf{i}_t^{c}$ is used to form beliefs. The typical solution method for the \textsc{aisg}s described above is dynamic programming aided by heuristic search \cite{saffidine23hsvi}. In generic \textsc{aisg}s, however, this solution method can not be applied due to belief hierarchies.  
\subsection{Subjective Rationality}
{
To avoid handling belief hierarchies, \textsc{col} takes inspiration from subjective rationality \cite{ryall03subjective}, which mandates the player to behave optimally given available information feedback and (possibly incorrect) subjective perception of the environment. In addition to Berk-Nash equilibrium (\Cref{def:berk-nash}) \cite{esponda16berk} considered in this work, many other equilibrium notions have been developed to characterize rationality under subjectivity. Among them, conceptually related to \textsc{col} are self-confirming equilibrium \cite{self-confirming, tao23sce} and subjective equilibrium \cite{kalai95subjective,arslan23subjective}, which require players' subjective conjecture to coincide with objective probability over realizable history. In contrast, \textsc{col} uses a relaxed consistency notion that only requires the equilibrium subjective conjecture to be the closest to the objective one among all admissible conjectures.     
}
\section{Conjectural Online Learning}
Conjectural online learning (\textsc{col}) is based on a forecaster-actor-critic (\textsc{fac}) architecture where the forecaster first conjectures the opponent's strategy $\widehat{\pi}_{-\K, t}$ to be used within a lookahead horizon $\ell_{\K}$ at each time step, against which the critic evaluates the player $\K$'s previous strategy $\pi_{\K,t-1}$ by estimating the value function $\widehat{J}_\K^{(\bm{\pi}_t)}$, $\bm{\pi}_t=(\pi_{\K,t-1}, \widehat{\pi}_{-\K,t})$. Finally, the actor updates the strategy through a $\ell_\K$-step rollout operation as
\begin{align}
        &\pi_{\K,t}(\mathbf{b}_t^\K)=\mathscr{R}(\widehat{\pi}_{-\K, t}, \mathbf{b}_t^\K, \widehat{J}_\K, \ell_\K)\triangleq \argmin_{a_t^{(\mathrm{k})}, a^{(\mathrm{k})}_{t+1},\hdots,a^{(\mathrm{k})}_{t+\ell_{\mathrm{k}}-1}} \label{eq:rollout} \\ &\mathbb{E}_{\bm{\pi}_t}\left[\sum_{j=t}^{t+\ell_{\mathrm{k}}-1}\gamma^{j-t}c^{\mathrm{k}}(S_j, A_j^{\K}, A_j^{-\K}) + \gamma^{\ell_{\mathrm{k}}} \widehat{J}_{\mathrm{k}}(\mathbf{B}^\K_{t+\ell_{\mathrm{k}}}) \mid \mathbf{b}^\K_t\right],\nonumber
\end{align}
where the random vector $\mathbf{B}^\K_{t+\ell_{\mathrm{k}}}$ denotes the reachable belief state at the future time step $t+\ell_\K$ under the strategy profile $\bm{\pi}_t$. Its realization depends on the actions, the conjecture $\widehat{\pi}_{-\K,t}$, and the Bayesian update \eqref{eq:belief-update}. 

When $\pi_{\K,t}$ is obtained from \eqref{eq:rollout}, the player executes the strategy, and the game moves to the next state, sending out new observations that update the players' beliefs through \eqref{eq:belief-update}. The forecaster then adapts the conjecture online using the information feedback $\mathbf{i}_t^\K$ through Bayesian learning, which ensures that its conjecture is asymptotically consistent with the observation. Subsequently, the actor and the critic repeat the rollout \eqref{eq:rollout} based on the updated conjecture $\widehat{\pi}_{-\K,t}$.

The pseudocode of \textsc{col} is listed in Alg.~\ref{alg:online_rollout}, and the main components of \textsc{col} are described below.
\begin{algorithm}
  \caption{\textbf{C}onjectural \textbf{O}nline \textbf{L}earning.}\label{alg:online_rollout}
\scriptsize
  \SetNoFillComment
  \SetKwProg{myInput}{Input:}{}{}
  \SetKwProg{myOutput}{Output:}{}{}
  \SetKwProg{myalg}{Algorithm}{}{}
  \SetKwProg{myproc}{Procedure}{}{}
  \SetKw{KwTo}{inp}
  \SetKwFor{Forp}{in parallel for}{\string do}{}%
  \SetKwFor{Loop}{Loop}{}{EndLoop}
  \DontPrintSemicolon
  \SetKwBlock{DoParallel}{do in parallel}{end}
  \myInput{
    \upshape Initial belief $\mathbf{b}_1$, game model $\Gamma$, base strategies $\bm{\pi}_{1}\triangleq (\pi_{\K,1}, \pi_{-\K,1})$, priors $(\mu^\K_1, \mu^{-\K}_1)$, lookahead horizons $\ell_{\mathrm{\K}},\ell_{-\mathrm{k}}$.
  }{}
  \myOutput{
    \upshape A sequence of action profiles $\mathbf{a}_{1},\mathbf{a}_{2}, \hdots$.
  }{}

  \myalg{}{
    \tcc{Initialization}
    $s_1 \sim \mathbf{b}_1$, $\mathbf{b}_1^{\K}=\mathbf{b}_1$, $\mathbf{h}^{\K}_1 \leftarrow (\mathbf{b}^\K_1)$, $\widehat{\pi}_{\K, 1} \leftarrow \pi_{\K, 1}$\;
    $a^{\K}_{1} \sim \pi_{\K,1}(\mathbf{b}^{\K}_1)$\;
    $s_{2} \sim f(\cdot \mid s_1, a^{\K}_{1}, a^{-\K}_{1})$\;
    \For{$t=2,3,\hdots$}{
      \tcc{Information feedback}
      $o^\K_t \sim z^\K(\cdot \mid s_t)$, $\mathbf{i}^{\K}_{t} \leftarrow (o^\K_{t}, a_{t-1}^\K)$, $\mathbf{h}^{\K}_{t} \leftarrow (\mathbf{h}^{\K}_{t-1}, \mathbf{i}^{\K}_{t})$\;
      $\mathbf{b}_{t} \leftarrow \mathscr{B}(\mathbf{b}_{t-1}^\K, \mathbf{i}_t^{\K}, \widehat{\pi}_{-\K,t-1})$ \;

        \tcc{Bayesian Forecaster}
      Update $\mu^{\K}_t$ via (\ref{eq:bayesian-learning}) and set $\widehat{\ell}_{-\K, t}\sim \mu_t^{\K}$\;

      \tcc{Conjectural Critic}
      Estimate $\widehat{J}_{-\K}^{(\pi_{\K,t-1}, \widehat{\pi}_{-\K,t-1})}$ \;
      Compute the conjectured opponent's strategy $\widehat{\pi}_{-\K, t}(\mathbf{b}^{\K}_t) \in \mathscr{R}(\pi_{\K, t-1}, \mathbf{b}_t^{\K}, \widehat{J}_{-\K}^{(\pi_{\K,t-1}, \widehat{\pi}_{-\K,t-1})}, \widehat{\ell}_{-\K, t})$\;
      Estimate $\widehat{J}_{\K}^{(\pi_{\K,t-1}, \widehat{\pi}_{-\K,t})}$\;

      \tcc{Actor}
      $\pi_{\K,t}(\mathbf{b}_t) \in \mathscr{R}(\widehat{\pi}_{-\K,t}, \mathbf{b}^{\K}_t, \widehat{J}_{\K}^{(\pi_{\K,t-1}, \widehat{\pi}_{-\K,t})}, \ell_{\K})$\;

    $a^{\K}_{t} \sim \pi_{\K,t}(\mathbf{b}^{\K}_t)$\;
    $s_{t+1} \sim f(\cdot \mid s_t, a^{\K}_{t}, a^{-\K}_{t})$
    }
  }
  \normalsize
\end{algorithm}

\subsection{Bayesian Forecaster and Consistent Conjecture}
The Bayesian forecaster in \textsc{col} begins with a prior probability measure $\mu^\K_1$ over a set of candidate opponent strategies $\widehat{\pi}_{-\K}\in \Pi_{-\K}\triangleq \{\pi| \pi: \mathcal{B}\rightarrow \Delta(\mathcal{A}^{-\K}) \}$. We assume that the opponent's strategy is parameterized by $\ell_{-\K}$, which means that it suffices to conjecture the parameter from $\K$'s candidate set: $\widehat{\ell}_{-\K}\in \Theta_\K$. We remark that the set of candidate strategies can be obtained from opponent modeling \cite{shen21opponent-modeling-review} or prior knowledge of the opponent. For example, suppose that both players employ rollout strategies, then the parameter set $\Theta_\K$ simply includes all possible lookahead horizons that the opponent could use (this is the parameterization used in Alg. \ref{alg:online_rollout}). Note that the candidate set may not include the actual opponent's strategy, meaning that the player's subjective modeling of its opponent can be misspecified \cite{esponda16berk}.

Upon receiving the information feedback $\mathbf{i}_t^{\K}$, the forecaster calculates the Bayesian posterior through \eqref{eq:bayesian-learning}, from which a new conjecture is sampled $\widehat{\ell}_{-\K, t}\sim \mu_t^\K$.
\begin{align}
\label{eq:bayesian-learning}
    \mu^{\K}_{t}({\widehat{\ell}})\triangleq \frac{\mathbb{P}[\mathbf{i}^{\mathrm{k}}_{t}\mid{\widehat{\ell}}, \mathbf{b}^\K_{t-1}]\mu_{t-1}^{\K}({\widehat{\ell}})}{\sum_{\overline{\ell} \in \Theta_{\mathrm{k}}}\mathbb{P}[\mathbf{i}^{\mathrm{k}}_{t}\mid\overline{\ell},\mathbf{b}^\K_{t-1}]\mu^{\K}_{t-1}(\widehat{\ell})},
\end{align}
where $\mathbb{P}[\mathbf{i}^{\mathrm{k}}_{t}\mid{\widehat{\ell}}, \mathbf{b}^\K_{t-1}]$ is the conditional probability of observing $\mathbf{i}_t^\K$ given the conjecture $\widehat{\ell}$ selected from the candidate set and the current belief state. For $\mathbf{i}^{\mathrm{k}}_{t}=\{o_t^\K, a_{t-1}^\K\}$, the conditional probability is given by
  \begin{align}
    &\mathbb{P}[\mathbf{i}^{\mathrm{k}}_{t}\mid{\widehat{\ell}}, \mathbf{b}^\K_{t-1}]=\sum_{a}\sum_{s, \tilde{s}} z^{\K}(o_t^\K|s)\widehat{\pi}_{-\K, t-1}(a_{t-1}^{-\K}=a|\mathbf{b}_{t-1}^\K)\nonumber \\
    &\times  \pi_{\K,t-1}(a_{t-1}^\K|\mathbf{b}_{t-1}^\K)f(s|\tilde{s}, a_{t-1}^\K, a_{t-1}^{-\K} ) \mathbf{b}_{t-1}^\K(\tilde{s}).   \label{eq:sub-cond-prod}
\end{align}
 (\ref{eq:bayesian-learning}) is well-defined under the following assumption. The first two are standard in parametric statistics \cite{esponda16berk}, while the last one ensures that players' conjecture sets are rich enough to make objective events also subjectively realizable.   
\begin{assumption}
    \label{assumption:bayes}
    (\textit{i}) $\Theta_{\K}$ is finite subset of an Euclidean space;  (\textit{ii}) $\mu_1^{\K}$ has full support; and (\textit{iii})  for all feasible $(\mathbf{i}^{\mathrm{k}},\mathbf{b}^{\K}_t)$, there exists $\widehat{\ell} \in \Theta_{\mathrm{k}}$ such that
    $\mathbb{P}[\mathbf{i}^{\K}|\widehat{\ell}, \mathbf{b}^{\K}_t]>0$.
\end{assumption}

One natural question arises: does the Bayesian posterior concentrate on the actual opponent's strategy when $\ell_{-\K}\in\Theta_{\K}$, achieving Bayesian consistency \cite{schwartz65consistency}? The answer is negative due to the belief bifurcation: $\mathbf{b}_t^\K \neq \mathbf{b}_t^{-\K}$. In \textsc{col}, each player $\K$ computes the conditional probability in \eqref{eq:sub-cond-prod} using its conjecture $\widehat{\pi}_{-\K,t-1}$ parameterized by $\widehat{\ell}_{-\K, t-1}$, which acts on its own belief state $\mathbf{b}_{t-1}^\K$. Hence, $\mathbb{P}[\mathbf{i}^{\mathrm{k}}_{t}\mid{\widehat{\ell}_{-\K,t-1}}, \mathbf{b}^\K_{t-1}]$ is a subjective conditional probability. In stark contrast, the objective conditional probability under the actual strategy $\pi_{-\K,t-1}$ (parameterized by $\ell_{-\K, t-1}$) is
    \begin{align}
    &\mathbb{P}[\mathbf{i}^{\mathrm{k}}_{t}\mid \ell_{-\K, t-1}, \mathbf{h}_t^\K\cup \mathbf{h}_t^{-\K}]=\mathbb{P}[\mathbf{i}^{\mathrm{k}}_{t}\mid \ell_{-\K, t-1}, \mathbf{b}^\K_{t-1}, \mathbf{b}^{-\K}_{t-1}]\nonumber \\
    &=\sum_{a}\sum_{s, \tilde{s}} z^{\K}(o_t^\K|s){\pi}_{-\K, t-1}(a_{t-1}^{-\K}=a|\mathbf{b}_{t-1}^{-\K}) \\
    &\times  \pi_{\K,t-1}(a_{t-1}^\K|\mathbf{b}_{t-1}^\K)f(s|\tilde{s}, a_{t-1}^\K, a_{t-1}^{-\K} ) \mathbf{b}_{t-1}^\K(\tilde{s}). \nonumber
\end{align}
Therefore, the Bayesian learning \eqref{eq:bayesian-learning} assigns more probability mass to conjectures under which $\mathbf{i}_t^\K$ is more likely to be observed under the player's subjective belief state, which may deviate from its opponent's. Consequently, even though some conjecture $\widehat{\ell}$ may well approximate the actual $\ell$, it won't be chosen if the strategy under the subjective belief state $\mathbf{b}_{t-1}^\K$ is less likely to induce $\mathbf{i}_t^\K$ than other parameters.

One may wonder what asymptotic behaviors \eqref{eq:bayesian-learning} displays if the traditional Bayesian consistency does not hold. We propose a new consistency metric based on the Kullback-Leibler (\textsc{kl}) divergence between the subjective and objective conditional probability. Referring to the parameter that minimizes the \textsc{kl} divergence as the consistent conjecture, we prove that the posterior concentrates on these consistent conjectures under Bayesian learning.

Denote by $\mathbf{b}_t\triangleq (\mathbf{b}_t^\K, \mathbf{b}_t^{-\K})\in \Delta(\mathcal{S})\times \Delta(\mathcal{S})\triangleq \mathcal{B}^2$ the joint belief state. Define a product measure $\bm{\nu}=(\nu^\K, \nu^{-\K})$ over the joint belief space, where $\nu^\K, \nu^{-\K}\in \Delta(\mathcal{S})$ are the occupancy measures. The \textsc{kl} divergence between the subjective and objective probability is defined as
\begin{equation}
    K(\widehat{\ell}_{-\K},\bm{\nu})\triangleq \E_{\mathbf{b}\sim \bm{\nu}}\E_{\mathbf{I}^\K}\left[\ln\left(\frac{\mathbb{P}[\mathbf{I}^{\mathrm{k}} \mid \ell_{-\K}, \mathbf{b}]}{\mathbb{P}[\mathbf{I}^{\mathrm{k}} \mid \widehat{\ell}_{-\K}, \mathbf{b}]}\right)\right],\label{eq:kl}
\end{equation}
where the information feedback $\mathbf{I}^\K$ follows the objective distribution $\mathbb{P}[\mathbf{I}^{\mathrm{k}} | \ell_{-\K}, \mathbf{b}]$. Note that (\ref{eq:kl}) uses the objective distribution as the reference and evaluates the deviation of the subjective $\mathbb{P}[\mathbf{I}^{\mathrm{k}} \mid \widehat{\ell}_{-\K}, \mathbf{b}]$. While we include the opponent's belief $\mathbf{b}^{-\K}$ in $\mathbb{P}[\mathbf{I}^{\mathrm{k}} \mid \widehat{\ell}_{-\K}, \mathbf{b}]$ for formality, it follows \eqref{eq:sub-cond-prod} and is independent of $\mathbf{b}^{-\K}$.

Given an occupancy measure $\bm{\nu}$, any conjecture $\widehat{\ell}$ that minimizes \eqref{eq:kl} is referred to as a \textit{consistent conjecture}. We denote the set of consistent conjectures by $\Theta_\K^\star(\bm{\nu})\triangleq \argmin_{\widehat{\ell}_{-k}\in \Theta_\K} K(\widehat{\ell}_{-\K},\bm{\nu})$, and the minimal divergence is denoted by $K_{\Theta_{\K}}^\star(\bm{\nu})$. Our consistency notion (\ref{eq:kl}) is different from Bayesian consistency, which requires that the posterior concentrates on the neighborhoods of the true parameter $\ell_{-\K}$\cite{schwartz65consistency}. As a relaxation, (\ref{eq:kl}) shifts focus from the parameter to the observation generation process: conjectures that induce subjective distributions closest to the objective distribution are said to be consistent, regardless of their distance to the true parameter. The following theorem states that the conjectures produced by Alg. \ref{alg:online_rollout} are asymptotically consistent with respect to the empirical occupancy measure $\bm{\nu}_{\mathbf{h}_t}\triangleq \frac{1}{t}\sum_{\tau=1}^{t}\mathds{1}_{\{\mathbf{b}\}}(\mathbf{b}_{\tau})$, and $\bm{\pi}_{\mathbf{h}_t}$ is the empirical strategy profile.
\begin{theorem}
\label{thm:asym-consistency}
    For any sequence $(\bm{\pi}_{\mathbf{h}_t}, \bm{\nu}_{\mathbf{h}_t})_{t\geq 1}$ from Alg.~\ref{alg:online_rollout},
    \begin{equation}
    \label{eq:asym-consistency}
        \lim_{t\rightarrow  \infty} \sum_{\widehat{\ell}_{-\K}\in\Theta_\K} (K(\widehat{\ell}_{-\K}, \bm{\nu}_{\mathbf{h}_t})-K^\star_{\Theta_\K}(\bm{\nu}_{\mathbf{h}_t}))\mu^\K_{t+1}(\widehat{\ell}_{-\K})=0,
    \end{equation}
    a.s.-$\mathbb{P}^{\mathscr{B},\mathscr{R}}$, where $\mathbb{P}^{\mathscr{B},\mathscr{R}}$ denotes the probability measure over the set of realizable histories $\mathbf{h}_t$ induced by $(\bm{\pi}_{\mathbf{h}_t})_{t\geq 1}$ under the rollout ($\mathscr{R}$) and Bayesian belief update ($\mathscr{B}$) in Alg.~\ref{alg:online_rollout}.
\end{theorem}

To see the consistency expressed by \eqref{eq:asym-consistency}, we first note that the difference in \eqref{eq:asym-consistency}, denoted by $\Delta K(\widehat{\ell}_{-\K}, \bm{\nu}_{\mathbf{h}_t})$, is non-negative. Therefore, the limit indicates that $\mu_t^\K$ assigns arbitrarily small probability mass to $\supp\{\Delta K(\widehat{\ell}_{-\K}, \bm{\nu}_{\mathbf{h}_t})\}$ and, equivalently, concentrates on $\Theta_\K^\star(\bm{\nu}_{\mathbf{h}_t})$ asymptotically. This observation leads to the following proof sketch.

\begin{proof}
For simplicity, we consider \textsc{col} with a binary candidate set $\Theta_{\K}=\{\widehat{\ell}_1, \widehat{\ell}_2\}$ and defer a rigorous proof to the appendix. Expressing $\mu_t^{\K}$ recursively using \eqref{eq:bayesian-learning} yields
\begin{align*}
 &\mu_{t+1}^{\K}(\widehat{\ell}_{1})
        = \left(1+ \rho_1\prod_{\tau=1}^{t}\frac{\mathbb{P}[\mathbf{i}^{\K}_{\tau+1} \mid \widehat{\ell}_{2},\mathbf{b}_{\tau}]}{\mathbb{P}[\mathbf{i}^{\K}_{\tau+1} \mid \widehat{\ell}_{1},\mathbf{b}_{\tau}]}\right)^{-1}\\
        &=\left(1+ \rho_1\prod_{\tau=1}^{t}\frac{\mathbb{P}[\mathbf{i}^{\K}_{\tau+1} \mid \widehat{\ell}_{2},\mathbf{b}_{\tau}]/\mathbb{P}[\mathbf{i}^{\K}_{\tau+1} \mid \ell_{-\K},\mathbf{b}_{\tau}]}{\mathbb{P}[\mathbf{i}^{\K}_{\tau+1} \mid \widehat{\ell}_{1},\mathbf{b}_{\tau}]/\mathbb{P}[\mathbf{i}^{\K}_{\tau+1} \mid \ell_{-\K},\mathbf{b}_{\tau}]}\right)^{-1},
\end{align*}
where $\rho_1=\mu_1^{\K}(\widehat{\ell}_2)/\mu_1^{\K}(\widehat{\ell}_1)$. The second equality holds as the numerator and the denominator are divided by the same objective conditional probability. Using the martingale convergence theorem, we obtain
\begin{align*}
    &\prod_{\tau=1}^{t}\frac{\mathbb{P}[\mathbf{i}^{\K}_{\tau+1} \mid \widehat{\ell}_{2},\mathbf{b}_{\tau}]}{\mathbb{P}[\mathbf{i}^{\K}_{\tau+1} \mid \ell_{-\K},\mathbf{b}_{\tau}]}= \exp\left\{ -t Z_{t+1}(\widehat{\ell}_2)\right\},\\
    &Z_{t+1}(\widehat{\ell}_2)\triangleq\left[\frac{1}{t}\sum_{\tau=1}^{t}\ln \frac{\mathbb{P}[\mathbf{i}^{\K}_{\tau+1} \mid \ell_{-\K},\mathbf{b}_{\tau}]}{\mathbb{P}[\mathbf{i}^{\K}_{\tau+1} \mid \widehat{\ell}_{2},\mathbf{b}_{\tau}]}\right]\xrightarrow[a.s.]{t\rightarrow\infty}K(\widehat{\ell}_2, \bm{\nu}_{\mathbf{h}_t}).
\end{align*}
Therefore, $\mu_{t+1}^{\K}(\widehat{\ell}_2)$ almost surely converges to
\begin{equation*}
    \left(1+\rho_1 \exp\left\{-t[K(\widehat{\ell}_{2}, \bm{\nu}_{\mathbf{h}_t})-K(\widehat{\ell}_{1}, \bm{\nu}_{\mathbf{h}_t})] \right\}\right)^{-1},
\end{equation*}
which approaches the Dirac-delta function on $\widehat{\ell_1}$, if $K(\widehat{\ell}_{2}, \bm{\nu}_{\mathbf{h}_t})-K(\widehat{\ell}_{1}, \bm{\nu}_{\mathbf{h}_t})>0$, i.e., concentrating on $\Theta_{\K}^\star$.
\end{proof}

After updating the conjecture $\widehat{\ell}_{-\K,t}$, player $\K$ reconstructs the opponent's strategy $\widehat{\pi}_{-\K,t}$ using the conjectured parameter. For example, suppose the parameter represents the conjectured opponent's lookahead horizon as shown in the case study. In this case, the player first performs an actor-critic update standing in the opponent's shoes, i.e., it performs a $\widehat{\ell}_{-\K,t}$-rollout against its own strategy $\pi_{\K,t-1}$. The resulting rollout policy $\widehat{\pi}_{-\K,t}$ serves as the conjectured strategy of the opponent (line 12 in Alg.~\ref{alg:online_rollout}), which further leads to the rollout update (line 14, Alg.~\ref{alg:online_rollout}). Such a conjectural rollout procedure is an instance of the proposed \textsc{fac} architecture in \textsc{col}.  One can freely incorporate deep learning methods into \textsc{col}, i.e., $\widehat{\pi}_{-\K}(\widehat{\ell}_{-\K})$ represents a generic parameterization.

\subsection{Equilibrium Analysis in Repeated \textsc{aisg}s}
Since \textsc{col} relies on a first-order belief of the private information, it is not sensible to discuss its connection to equilibrium concepts in \textsc{aisg}s that involve the universal type space based on infinite belief hierarchies \cite[Chapter 11]{solan_game}. We instead relate the asymptotics of Alg.~\ref{alg:online_rollout} to the Berk-Nash equilibrium, a recently popularized concept characterizing players' rational behaviors under their \textit{subjective} perceptions of the game \cite{esponda16berk}. To streamline the analysis, we focus on a special case of \textsc{aisg}: repeated games with stochastic states and observable actions. In this type of \textsc{aisg}, the state $s_t$ is sampled from $\mathbf{b}_1$ repeatedly at each stage, which ensures that the occupancy measure $\nu^{\K}$ exists and is uniquely determined by $\mathbf{b}_1$ and the kernel $z^{k}$ \cite{tao23pot}.

Compared with the \textsc{aisg} setup in \Cref{sec:aisg}, the repeated game assumes public observations on actions rather than merely partial observations, i.e., $(a_{t-1}^{\K}, a_{t-1}^{-\K})$ is revealed to the players prior to their decisions at time $t$ and $\mathbf{i}_t^{\K}=\{a_{t-1}^{\K}, a_{t-1}^{-\K}, o_t^{\K}\}$. This information structure secures the subjective conditional probability's dependence on the opponent, leading to an effective Bayesian forecast.

Once player $\K$ observes $o_t^{\K}$, its posterior belief is $\mathbf{b}_t^{\K}(s_t=s)=\mathbf{b}_1(s)z^{\K}(o_t^{\K})/\sum_{\tilde{s}}\mathbf{b}_1(\tilde{s})z^{\K}(o_t^{\K}|\tilde{s})$, which is solely determined by its private observation. Hence, the behavior strategy for the repeated game is a mapping $\pi_{\K}: \mathcal{O}^{\K}\rightarrow \Delta(\mathcal{A}^{\K})$. When the player adopts Alg.~\ref{alg:online_rollout}, the subjective conditional probability in \eqref{eq:sub-cond-prod} (suppressing $o_t^{\K}, a_{t-1}^{\K}$) turns into $ \mathbb{P}[a_{t-1}^{-\K}\mid \widehat{\ell}_{-\K, t-1}, \mathbf{b}_{t-1}^{\K}]=\sum_{s\in \mathcal{S}, o\in \mathcal{O}^{-\K}}\mathbf{b}_{t-1}^{\K}(s)z^{-\K}(o|s)\widehat{\pi}_{-\K,t-1}(o)$, which leads to the Bayesian update in \eqref{eq:bayesian-learning}.


Following the \textsc{col} setup, each player acquires a private candidate set $\Theta_\K$ that prescribes its subjective perceptions of potential opponent strategies, also known as subjective modeling \cite{esponda16berk}. We assume that players are myopic and minimize the expected stage cost at each time step. Consequently, the rollout operation \eqref{eq:rollout} reduces to the best response dynamics against the conjectured opponent: $\pi_{\K,t}\in \argmin_{\pi} \E_{\pi, \widehat{\pi}_{-\K,t}}[c^{\K}(S_t,A_t^{\K}, A_t^{-\K})\mid \mathbf{b}_t^{\K}]$, where $\ell_{\K}=1$. Such a conjectured best response can also be written as $\pi_{\K,t}\in \argmin_{\pi} \E_{ \mathbb{P}[\cdot\mid \widehat{\ell}_{-\K,t}, \mathbf{b}_t^{\K}]}[c^{\K}(S_t,A_t^{\K}, A_t^{-\K})\mid \mathbf{b}_t^{\K}]$. Denote by $\mathbb{P}[\cdot\mid \mu^{\K}, \mathbf{b}^{\K}]\triangleq \sum_{\widehat{\ell}}\mu^{\K}(\widehat{\ell}_{-\K})\mathbb{P}[\cdot\mid \widehat{\ell}_{-\K}, \mathbf{b}^{\K}]$.


The pair $\langle \Gamma, \{\Theta_{\K}\}_{\K\in \mathcal{N}}\rangle$ leads to a game where players have distinct subjective perceptions, which is beyond classical game-theoretic solution concepts (e.g., Nash equilibrium), which assume that all players know $\Gamma$. A more appropriate solution concept for this type of game is the Berk-Nash equilibrium (\Cref{def:berk-nash}), which characterizes a steady state where each player follows the best response based on its subjective conjecture (optimality) and where conjectures are consistent with objective observations (consistency).
\begin{definition}[Berk-Nash Equilibrium, adapted from \cite{esponda16berk}]
\label{def:berk-nash}
    A strategy profile $(\pi_{\K}, \pi_{-\K})$, $\pi_{\K}:\mathcal{O}^{\K}\rightarrow \Delta(\mathcal{A}^{\K})$, is a Berk-Nash equilibrium of $\langle \Gamma, \{\Theta_{\K}\}_{\K\in \mathcal{N}}\rangle$ if there exists $\mu^\K\in \Delta(\Theta_{\K})$, for all $\K\in\mathcal{N}$, such that
    \begin{enumerate}[(i),leftmargin=*]
        \item (\textsc{optimality}) for any $o^{\K}$ and its induced belief $\mathbf{b}^{\K}$, $\pi_{\K}(\cdot|o^{\K})\in \argmin_{\pi}\E_{\mathbb{P}[\cdot\mid \mu^{\K}, \mathbf{b}^{\K}]}[c^{\K}(S, A^{\K}, A^{-\K})]$,
        \item (\textsc{consistency}) $\mu^{\K}\in \Delta(\Theta_{\K}^\star)$.
    \end{enumerate}
  \end{definition}

\vspace{2mm}

If the sequence $(\bm{\pi}_{\mathbf{h}_t})_{t\geq 1}$ produced by \textsc{col} (Alg.~\ref{alg:online_rollout}) converges, then \Cref{thm:asym-consistency} asserts the consistency condition while the best response ensures the optimality condition in \Cref{def:berk-nash}, as stated in Cor.\ref{coro:convergence-berk-nash}. We leave the last-iterate convergence analysis, e.g., \cite{shutian23erm}, as the future extension.
\begin{corollary}
\label{coro:convergence-berk-nash}
If  $(\bm{\pi}_{\mathbf{h}_t})_{t \geq 1}$ by Alg.~\ref{alg:online_rollout} converges, then it converges to a Berk-Nash equilibrium of $\langle \Gamma, \{\Theta_{\K}\}_{\K\in \mathcal{N}}\rangle$.
\end{corollary}
\begin{proof}
    We present the key steps in this sketch. Due to the i.i.d. states, the belief state $\mathbf{b}_t^{\K}$ is also i.i.d., and hence, the empirical occupancy measure converges almost surely by the Glivenko-Cantelli theorem. Denote by $\bm{\nu}$ the limit point of the joint belief occupancy. From \Cref{thm:asym-consistency}, $\mu_t^\K$ asymptotically concentrates on $\Theta_{\K}^\star(\bm{\nu})$. Note that $\Delta(\Theta_{\K}^\star(\bm{\nu}))$, as a finite-dimensional probability simplex, is compact. There exists a convergent subsequence of $(\mu_t^{\K})_{t\geq 1}$ whose limit point is denoted by $\mu_{\infty}^{\K}\in \Delta(\Theta_{\K}^\star)$ and is the consistent conjecture distribution in (ii) in \Cref{def:berk-nash}, and the rest is to show that the limit point of the empirical strategy profile satisfies (i). Since $\pi_t^\K$ is induced by the best response correspondence, we invoke the Painlev\'e-Kuratowski set convergence theorem \cite[Thm 7.11]{RockWets98} to complete the proof. 
\end{proof}

\begin{figure*}
  \centering
  \scalebox{0.7}{
    \input{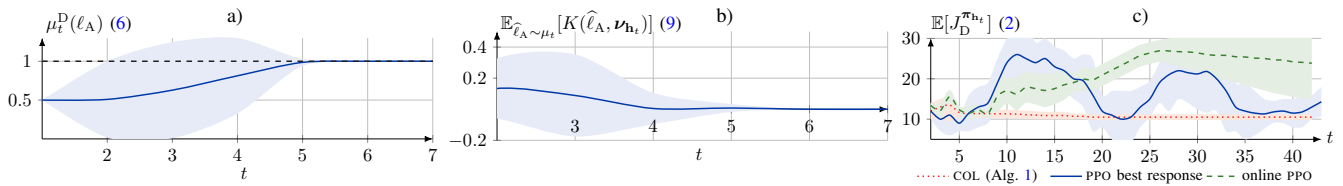}
  }
  \caption{Evaluation results for the intrusion response case study; values indicate the mean; the shaded areas and the error bars indicate the 95\% confidence interval based on $20$ random seeds; hyperparameters are listed in online appendix.}
  \label{fig:evaluation_results}
  \vspace{-0.5cm}
\end{figure*}
\section{Case Study: Intrusion Response}\label{sec:case_study}
We present \textsc{col} through a use case on defending the IT infrastructure of an organization against Advanced Persistent Threats (\textsc{apt}s). We formulate this use case as a zero-sum partially observable stochastic game between the defender and the attacker (a type of \textsc{aisg}). The state $s_t$ represents the number of compromised servers, and the initial distribution is a point mass $\mathbf{b}_1(s=0)=1$. Both players can invoke two actions: ($\mathsf{S}$)top and ($\mathsf{C}$)ontinue. $\mathsf{S}$ triggers a change in the game state while $\mathsf{C}$ is a passive action that does not change the state. Specifically, $a^{\mathrm{A}}_t=\mathsf{S}$ is the attacker's compromise action and $a^{\mathrm{D}}_t=\mathsf{S}$ is the defender's recovery action. The state transitions follow $f(S_{t+1}=0 \mid s_{t}, a^{\D}_t=\mathsf{S}, a^{\mathrm{A}}_t) = 1$, $f(S_{t+1}=s_t \mid s_{t}, \mathsf{C}, \mathsf{C}) = 1$, $f(S_{t+1}=s_t \mid s_{t}, \mathsf{C}, \mathsf{S}) =0$, and $f(S_{t+1}=\min[s_t+1,N]  \mid s_{t}, \mathsf{C}, \mathsf{S}) =1 $. The game involves asymmetric information as the defender is uncertain about the compromised servers. The information structures are $\mathbf{i}_t^{\A}=(a^{\D}_{t-1}, a^{\A}_{t-1}, o_t, s_t)$ and $\mathbf{i}_t^{\D}=(a_{t-1}^{\D}, o_t)$, where $o_t$ represents the number of system alerts. The distribution of $o_t$ is estimated based on measurement data from a digital twin (see online appendix and \cite{csle_source_code}). The performance of a defender strategy is quantified through the cost function:
\begin{align*}
c(s_t, a^{(\mathrm{D})}_t) &\triangleq \overbrace{s^{5/4}_t\mathds{1}_{a^{(\mathrm{D})}_t\neq\mathsf{S}}}^{\text{intrusion cost}} + \overbrace{\mathds{1}_{a^{(\mathrm{D})}_t=\mathsf{S}}(1- 2\mathds{1}_{s_t>0})}^{\text{response action cost}},
\end{align*}
which encourages rapid response when servers are compromised ($s_t^{5/4}$ is strictly increasing) while penalizing unnecessary recoveries.

We implement the use case on an infrastructure with $64$ servers where we run $10$ different types of real-world \textsc{apt}s. The infrastructure configuration and the \textsc{apt} instances are adapted from \cite{tifs_hlsz_supplementary} and presented in the appendix. The attacker also employs online rollout (nonstationary attack), and the lookahead horizon is $\ell_{\A}=1$. The candidate set is $\Theta=\{1,2\}$ with $\mu_1^{\K}$ being the uniform distribution.  Initial (belief-independent) strategies are given by $\pi_{\D,1}(\mathsf{S}|\cdot)=1$ and $\pi_{\A,1}(\mathsf{S}|\cdot)=0.05$. \textsc{ppo} baselines setup is in the appendix.

Figures \ref{fig:evaluation_results}.a--b show the evolution of the defender's conjecture distribution $\mu^{\mathrm{D}}_t$ (\ref{eq:bayesian-learning}) and the discrepancy (\ref{eq:kl}). We observe that $\mu_t$ converges and concentrates on the consistent conjecture after $5$ time steps, as predicted by \Cref{thm:asym-consistency}. Figure \ref{fig:evaluation_results}.c shows the expected cost (\ref{eq:objective}) of Alg. \ref{alg:online_rollout} and the expected cost of two reinforcement learning techniques: online self-play with \textsc{ppo} \cite[Alg. 1]{ppo} and approximate best response dynamics with \textsc{ppo} \cite[Alg. 1]{ppo}. We note that the expected cost of best response dynamics oscillates while online \textsc{ppo} does not converge. Similar behavior of best response dynamics has been observed in related work \cite{tifs_hlsz_supplementary}. By contrast, the expected cost of Alg. \ref{alg:online_rollout} is significantly more stable, and its behavior is consistent with convergence to a Berk-Nash equilibrium.

\section{Conclusion}
We propose conjectural online learning (\textsc{col}), an online method for stochastic games with information asymmetry. By utilizing first-order beliefs, \textsc{col} spares learning agents from nested beliefs, ensures conjecture consistency with information feedback, and adapts efficiently to the conjectured opponent. The empirical strategy profile of \textsc{col}, shall it stabilize, converges to the Berk-Nash equilibrium. The numerical results indicate the fast convergence of \textsc{col} against online-rollout-based attacks. A promising future direction for this online learning paradigm is to design and influence players' subjective perceptions for incentive provision purposes.
\bibliographystyle{ieeetr}
\bibliography{ref.bib}

\newpage
\begin{appendix}

\subsection{Full Proof of \Cref{thm:asym-consistency}}
\setcounter{equation}{0}
\renewcommand{\theequation}{\Alph{section}.\arabic{equation}}
Following the proof sketch in the main text, our proof begins with expressing $\mu_t^\K$ in terms of log-likelihood ratios. The sample average of the log-likelihood fractions almost surely converges to the KL divergence by the martingale convergence theorem \cite[Thm. 6.4.3]{Ash:2000uj}. Towards this proof, we state the following lemma.
\begin{lemma}
\label{app-lem:mds}
    For any $\widehat{\ell}_{-\K}\in \Theta_{\K}$ and $(\bm{\pi}_{\mathbf{h}_t}, \bm{\nu}_{\mathbf{h}_t})$ generated by Alg.~\ref{alg:online_rollout}, the following limit holds a.s.-$\mathbb{P}^{\mathscr{B},\mathscr{R}}$,
    \begin{equation*}
        \lim_{t\rightarrow 0 }\bigg| \underbrace{\frac{1}{t}\sum_{\tau=1}^t \ln \frac{\mathbb{P}[\mathbf{i}^{\K}_{\tau+1} \mid \ell_{-\K}, \mathbf{b}_\tau]}{\mathbb{P}[\mathbf{i}^{\K}_{\tau+1} \mid \widehat{\ell}_{-\K}, \mathbf{b}_\tau]}}_{\triangleq Z_{t+1}(\widehat{\ell}_{-\K})}-K(\widehat{\ell}_{-\K}, \bm{\nu}_{\mathbf{h}_t})\bigg|=0.
    \end{equation*}
\end{lemma}
\begin{proof}
    By the definition of $Z_{t+1}$ and $\bm{\nu}_{\mathbf{h}_t}$,
\begin{align*}
        &Z_{t+1}(\widehat{\ell}_{-\K}) = \sum_{\mathbf{b} \in \mathcal{B}^2}\sum_{\tau=1}^{t}t^{-1}\mathds{1}_{\{\mathbf{b}\}}(\mathbf{b}_{\tau})\ln \frac{\mathbb{P}[\mathbf{i}^{\K}_{\tau+1} \mid \ell_{-\K}, \mathbf{b}_\tau]}{\mathbb{P}[\mathbf{i}^{\K}_{\tau+1} \mid \widehat{\ell}_{-\K}, \mathbf{b}_\tau]}\\
&=\mathbb{E}_{\mathbf{b}\sim \bm{\nu}_{\mathbf{h}_t}}\left[\sum_{\tau=1}^{t}\frac{\ln \mathbb{P}[\mathbf{i}^{\K}_{\tau+1} \mid \ell_{-\K}, \mathbf{b}]}{t} - \sum_{\tau=1}^{t}\frac{\ln \mathbb{P}[\mathbf{i}^{\K}_{\tau+1} \mid \widehat{\ell}_{-\K}, \mathbf{b}]}{t}\right]
\end{align*}
It suffices to prove that $t^{-1}\sum_{\tau=1}^t \ln \mathbb{P}[\mathbf{i}^{\K}_{\tau+1} \mid \ell_{-\K}, \mathbf{b}_\tau]$ converges to $\E_{\mathbf{b}\sim \bm{\nu}_{\mathbf{h}_t}}\mathbb{E}_{\mathbf{I}^{\K}}[\ln \mathbb{P}[\mathbf{I}^{\K} | \overline{\ell}_{-\K},\mathbf{b}]$ almost surely, as the second sum inside the expectation share the same proof. As a reminder, the expectation $\mathbb{E}_{\mathbf{I}^{\K}}[\ln \mathbb{P}[\mathbf{I}^{\K} | \overline{\ell}_{-\K},\mathbf{b}]$ is taken with respect to the objective distribution $\mathbb{P}[\mathbf{I}^{\K} | \overline{\ell}_{-\K},\mathbf{b}]$.

We first show that $(\bm{\pi}_{\mathbf{h}_t})_{t\geq 1}$ generated by Alg. \ref{alg:online_rollout} induces a well-defined probability measure $\mathbb{P}^{\mathscr{B}, \mathscr{R}}$ over the set of realizable histories $\mathbf{h} \in \mathcal{H}^{\K}\times \mathcal{H}^{-\K}$. Since the space of realizable histories $\mathbf{h}\in \mathcal{H}^{\K}\times \mathcal{H}^{-\K}$ is a product of finite measurable spaces, it is countable. The strategies are Markovian with respect to the belief, and hence, the Ionescu-Tulcea extension theorem \cite{Ash:2000uj} asserts the existence of a probability measure over $\mathcal{H}^{\K}\times \mathcal{H}^{-\K}$, which is denoted by $\mathbb{P}^{\mathscr{B}, \mathscr{R}}$.

We now prove the almost sure convergence. Let $X_{\tau} \triangleq \ln\mathbb{P}[\mathbf{i}^{\K}_{\tau+1} | \ell_{-\K}, \mathbf{b}_{\tau}] - \mathbb{E}_{\mathbf{I}^{\K}}\left[\ln\mathbb{P}[\mathbf{I}^{\K} | \ell_{-\K}, \mathbf{b}_{\tau}]\right]$, and $(X_{\tau})_{\tau\geq 1}$ is a martingale difference sequence (\textsc{mds}). To see this, we need to prove that (a) $\mathbb{E}[X_{\tau} | \mathbf{h}_{\tau-1}]=0$; and (b) $\mathbb{E}[|X_{\tau}|] < \infty$. By definition,
\begin{align*}
    &\mathbb{E}[X_{\tau} | \mathbf{h}_{\tau-1}]\\
&=\mathbb{E}_{\mathbf{I}^{\K}}\left[\ln\mathbb{P}[\mathbf{I}^{\K} |\ell_{-\K}, \mathbf{b}_{\tau}] - \mathbb{E}_{\mathbf{I}^{\K}}\left[\ln\mathbb{P}[\mathbf{I}^{\K} | \ell_{-\K}, \mathbf{b}_{\tau}]\right] | \mathbf{h}_{\tau-1}\right]\\
&=\mathbb{E}_{\mathbf{I}^{\K}}\left[\ln\mathbb{P}[\mathbf{I}^{\K} | \ell_{-\K}, \mathbf{b}_{\tau}]\right] -\mathbb{E}_{\mathbf{I}^{\K}}\left[\ln\mathbb{P}[\mathbf{I}^{\K} | \ell_{-\K}, \mathbf{b}_{\tau}]\right]=0,
\end{align*}
where the second equality follows the fact that $\mathbf{I}^{\K}$ is conditionally independent of $\mathbf{h}_{\tau-1}$ given $\mathbf{b}_{\tau}$. This completes the proof of (a).

To prove (b), we start by applying Jensen's inequality:
\begin{equation}
    \mathbb{E}[|X_{\tau}|]=(\mathbb{E}[|X_{\tau}|]^2)^{1/2}\leq (\mathbb{E}[X_{\tau}^2])^{1/2},\label{eq:jensen}
\end{equation}
and hence, it suffices to bound $\E[X_{\tau}]$. Toward this end, we rewrite $\ln\mathbb{P}[\mathbf{i}^{\K}_{\tau+1} |\ell_{-\K}, \mathbf{b}_{\tau}]$ using the fact that $\mathbb{P}[\mathbf{i}^{\K}_{\tau+1}\mid \ell_{-\K}, \mathbf{b}_{\tau}]=\mathbb{E}_{\mathbf{I}^{\K}}[\mathds{1}_{\mathbf{I}^{\K}}(\mathbf{i^{\K}_{\tau+1}})]$:
\begin{align*}
    \ln\mathbb{P}[\mathbf{i}^{\K}_{\tau+1} |\ell_{-\K}, \mathbf{b}_{\tau}] = \frac{\mathbb{E}_{\mathbf{I}^{\K}}\left[\mathds{1}_{\{\mathbf{I}^{\K}\}}(\mathbf{i}^{\K}_{\tau+1})\ln\mathbb{P}[\mathbf{I}^{\K} \mid \ell_{-\K}, \mathbf{b}_{\tau}]\right]}{\mathbb{P}[\mathbf{i}_{\tau+1}^{\K} \mid \ell_{-\K}, \mathbf{b}_{\tau}]}.
\end{align*}
This new expression rewrites $X_\tau$ as
\begin{equation}
\label{eq:x_tau_1}
X_{\tau}=\frac{\E_{\mathbf{I}^{\K}}\left[\kappa \left(\mathds{1}_{\{\mathbf{I}^{\K}\}}(\mathbf{i}^{\K}_{\tau+1}) - \mathbb{P}[\mathbf{i}_{\tau+1}^{\K} | \ell_{-\K}, \mathbf{b}_{\tau}]\right) \right]}{\mathbb{P}[\mathbf{i}^{\K}_{\tau+1}\mid \ell_{-\K}, \mathbf{b}_{\tau}]},
\end{equation}
where $\kappa\triangleq \ln\mathbb{P}[\mathbf{I}^{\K} | \ell_{-\K}, \mathbf{b}_{\tau}]$.

For realizable histories, $\mathbb{P}[\mathbf{i}_{\tau+1}^{\K} | \ell_{-\K}, \mathbf{b}_{\tau}] \in (0,1]$, and thus, it is safe to ignore the denominator in \eqref{eq:x_tau_1} and only consider the upper-bound of the numerator.  Applying the Cauchy-Schwartz inequality to the numerator, we obtain
\begin{equation*}
    X_\tau^2 \leq \mathbb{E}_{\mathbf{I}^{\K}}\Big[\kappa^2 \big(\underbrace{\mathds{1}_{\{\mathbf{I}^{\K}\}}(\mathbf{i}^{\K}_{\tau}+1) - \mathbb{P}[\mathbf{i}_{\tau+1}^{\K} | \ell_{-\K}, \mathbf{b}_{\tau}]}_{\triangleq \chi}\big)^2 \Big].
\end{equation*}
Since $\kappa \geq 0$ and $\chi \in [0,1]$, $X^2_{\tau} \leq \mathbb{E}_{\mathbf{I}^{\K}}\left[\kappa^2 \chi^2\right] \leq \mathbb{E}_{\mathbf{I}^{\K}}\left[\kappa^2 \right]$. Since the mapping $x\mapsto x(\ln x)^2$ is bounded by 1 for $x\in (0,1]$\footnote{we use the standard convention that $0(\ln 0)^2=0$}, $\E_{\mathbf{I}^{\K}}[\kappa^2]\leq 1$, which proves that $\E[X_\tau^2]\leq 1$ and asserts that $(X_\tau)_{\tau\geq 1}$ is a \textsc{mds}.

Define a martingale $(Y_t)_{t\geq 1}$, $Y_{t}\triangleq \sum_{\tau=1}^t X_\tau/\tau$. By the martingale convergence theorem, $(Y_{\tau})_{\tau \geq 1}$ converges to a finite and integrable random variable a.s.-$\mathbb{P}^{\mathscr{B}, \mathscr{R}}$ \cite[Thm. 6.4.3]{Ash:2000uj}. This convergence enables us to invoke Kronecker's lemma, stating that $\lim_{t \rightarrow \infty}t^{-1}\sum_{\tau=1}^tX_{\tau}=0$ a.s.-$\mathbb{P}^{\mathscr{R}}$ \cite[pp. 105]{pollard_2001}. Consequently, the following difference converges to zero a.s.-$\mathbb{P}^{\mathscr{B}, \mathscr{R}}$ as $t \rightarrow \infty$
\begin{align*}
\bigg| \frac{1}{t}\sum_{\tau=1}^t\frac{\ln\mathbb{P}[\mathbf{i}^{\K}_{\tau+1} | \ell_{-\K}, \mathbf{b}_{\tau}]}{t} - \mathbb{E}_{\mathbf{b} \sim \nu_t}\mathbb{E}_{\mathbf{I}^{\K}}\left[\ln\mathbb{P}[\mathbf{I}^{\K} | \ell_{-\K}, \mathbf{b}]\right]\bigg|.
\end{align*}
\end{proof}

With \Cref{app-lem:mds}, we unfold the complete proof of \Cref{thm:asym-consistency} as below.

\begin{proof}
   Following the proof sketch in the main text, by recursively applying the Bayes rule \eqref{eq:bayesian-learning}, we obtain that for $\widehat{\ell}_{-\K}\in \Theta_{\K}$ ($\overline{\ell}_{-\K}$ denotes a generic candidate in $\Theta_{\K}$),
   \begin{align*}
       \mu_{t+1}^{\K}(\widehat{\ell}_{-\K})=\frac{\mu^{\K}_1(\widehat{\ell}_{-\K})\exp\left(-tZ_{t+1}(\widehat{\ell}_{-\K})\right)}{\sum_{\Theta_{\K}}\mu^{\K}_1(\overline{\ell}_{-\K})\exp\left(-tZ_{t+1}(\overline{\ell}_{-\K})\right)},
   \end{align*}
   where $Z_{t+1}$ is defined in \Cref{app-lem:mds}. The expectation in \eqref{eq:asym-consistency} can then be rewritten as
   \begin{align}
       &\sum_{\widehat{\ell}_{-\K}\in\Theta_\K} \overbrace{(K(\widehat{\ell}_{-\K}, \bm{\nu}_{\mathbf{h}_t})-K^\star_{\Theta_\K}(\bm{\nu}_{\mathbf{h}_t}))}^{\triangleq \Delta K(\widehat{\ell}_{-\K}, \bm{\nu}_{\mathbf{h}_t})}\mu^\K_{t+1}(\widehat{\ell}_{-\K})\nonumber  \\
       =&\frac{\sum_{\widehat{\ell}_{-\K}}\Delta K(\widehat{\ell}_{-\K}, \bm{\nu}_{\mathbf{h}_t})\mu^{\K}_1(\widehat{\ell}_{-\K})\exp\left(-tZ_{t+1}(\widehat{\ell}_{-\K})\right)}{\sum_{\overline{\ell}_{-\K}}\mu^{\K}_1(\overline{\ell}_{-\K})\exp\left(-tZ_{t+1}(\overline{\ell}_{-\K})\right)} \nonumber\\
       =&\frac{\sum_{\widehat{\ell}_{-\K}}\Delta K(\widehat{\ell}_{-\K}, \bm{\nu}_{\mathbf{h}_t})\mu^{\K}_1(\widehat{\ell}_{-\K})\exp\left(-t\left(Z_{t+1}(\widehat{\ell}_{-\K})-K_{\Theta_{\K}}^\star(\bm{\nu}_{\mathbf{h}_t})\right)\right)}{\sum_{\overline{\ell}_{-\K}}\underbrace{\mu^{\K}_1(\overline{\ell}_{-\K})\exp\left(-t\left(Z_{t+1}(\overline{\ell}_{-\K})-K_{\Theta_{\K}}^\star(\bm{\nu}_{\mathbf{h}_t})\right)\right)}_{\triangleq \sigma_{t+1}(\overline{\ell}_{-\K})}},
       \label{eq:expected-kl}
   \end{align}
where the second equality is obtained by multiplying the numerator and denominator by $\exp(tK_{\Theta_{\K}}^\star(\bm{\nu}_{\mathbf{h}_t}))$.

   For an arbitrarily small $\epsilon>0$, $\Theta_{\K}^\epsilon\triangleq \{\widehat{\ell}|\Delta K(\widehat{\ell}_{-\K}, \bm{\nu}_{\mathbf{h}_t}) \geq \epsilon \}$. The following inequalities hold for the fraction in \eqref{eq:expected-kl}:
   \begin{align}
      &\frac{\sum_{\widehat{\ell}_{-\K}\in \Theta_{\K}}\Delta K(\widehat{\ell}_{-\K}, \bm{\nu}_{\mathbf{h}_t})\sigma_{t+1}(\widehat{\ell}_{-\K})}{\sum_{\overline{\ell}_{-\K}\in \Theta_{\K}}\sigma_{t+1}(\overline{\ell}_{-\K})}\nonumber\\
      =& \frac{\left(\sum_{\Theta_{\K}^{\epsilon}}+\sum_{\Theta_{\K}\setminus \Theta_{\K}^{\epsilon}}\right)\Delta K(\widehat{\ell}_{-\K}, \bm{\nu}_{\mathbf{h}_t})\sigma_{t+1}(\widehat{\ell}_{-\K})}{\sum_{\overline{\ell}_{-\K}}\sigma_{t+1}(\overline{\ell}_{-\K})}\nonumber\\
      \leq & \epsilon + \frac{\sum_{\widehat{\ell}_{-\K}\in \Theta_{\K}^{\epsilon}}\Delta K(\widehat{\ell}_{-\K}, \bm{\nu}_{\mathbf{h}_t})\sigma_{t+1}(\widehat{\ell}_{-\K})}{\sum_{\overline{\ell}_{-\K}\in \Theta_{\K}}\sigma_{t+1}(\overline{\ell}_{-\K})} \nonumber.
   \end{align}
   Hence, it suffices to prove that
   \begin{align}
   \label{eq:upper-bound convergence}
       \lim_{t\rightarrow \infty}\frac{\sum_{\widehat{\ell}_{-\K}\in \Theta_{\K}^{\epsilon}}\Delta K(\widehat{\ell}_{-\K}, \bm{\nu}_{\mathbf{h}_t})\sigma_{t+1}(\widehat{\ell}_{-\K})}{\sum_{\overline{\ell}_{-\K}\in \Theta_{\K}}\sigma_{t+1}(\overline{\ell}_{-\K})}= 0
   \end{align}

   Towards the proof of \eqref{eq:upper-bound convergence}, we rewrite  the exponent in $\sigma_{t+1}(\overline{\ell}_{-\K})$ as
   \begin{align}
&-t(Z_{t+1}(\overline{\ell}_{-\K})-K^{\star}_{\Theta_{\K}}(\bm{\nu}_{\mathbf{h}_t}))=\nonumber\\
&-t(Z_{t+1}(\overline{\ell}_{-\K})-K^{\star}_{\Theta_{\K}}(\bm{\nu}_{\mathbf{h}_t})) + K(\overline{\ell}_{-\K},\bm{\nu}_{\mathbf{h}_t}) -K(\overline{\ell}_{-\K},\bm{\nu}_{\mathbf{h}_t})\nonumber\\
&=-t(\Delta K(\overline{\ell}_{-\K}, \bm{\nu}_{\mathbf{h}_t}) + Z_{t+1}(\overline{\ell}_{-\K}) -K(\overline{\ell}_{-\K},\bm{\nu}_{\mathbf{h}_t})). \label{eq:numerator_wrangling_thm_4}
   \end{align}
Recall from \Cref{app-lem:mds} that for any $\epsilon>0$, there exists $\eta>0$ and $t_{\eta}\geq 1$ such that $|Z_{t+1}(\overline{\ell}_{-\K})-K(\overline{\ell}_{-\K},\bm{\nu}_{\mathbf{h}_t} )|< \eta $ for all $t\geq t_{\eta}$ and $\overline{\ell}_{-\K}\in \Theta_{\K}$. Such a $t_{\eta}$ is uniform since $\Theta_{\K}$ is  finite (\Cref{assumption:bayes}).  Consequently, the fraction in \eqref{eq:upper-bound convergence} satisfies the following inequality for all $t\geq t_{\eta}$:
\begin{align}
    &\frac{\sum_{\widehat{\ell}_{-\K}\in \Theta_{\K}^{\epsilon}}\Delta K(\widehat{\ell}_{-\K}, \bm{\nu}_{\mathbf{h}_t})\sigma_{t+1}(\widehat{\ell}_{-\K})}{\sum_{\overline{\ell}_{-\K}\in \Theta_{\K}}\sigma_{t+1}(\overline{\ell}_{-\K})}\tag{$\ast$} \\
    \leq & \frac{\sum_{\Theta_{\K}^\epsilon}\Delta K(\widehat{\ell}_{-\K}, \bm{\nu}_{\mathbf{h}_t}) \mu^{\K}_1(\widehat{\ell}_{-\K}) \exp\left(-t(\Delta K(\widehat{\ell}_{-\K}, \bm{\nu}_{\mathbf{h}_t}) -\eta)\right) }{\sum_{\Theta_{\K}}\mu_1^{\K}(\overline{\ell}_{-\K})\exp\left(-t(\Delta K(\overline{\ell}_{-\K}, \bm{\nu}_{\mathbf{h}_t}) +\eta)\right)   }\nonumber\\
    = & e^{2t\eta} \frac{\sum_{\Theta_{\K}^\epsilon}\mu^{\K}_1(\widehat{\ell}_{-\K}) \Delta K(\widehat{\ell}_{-\K}, \bm{\nu}_{\mathbf{h}_t})  \exp\left(-t\Delta K(\widehat{\ell}_{-\K}, \bm{\nu}_{\mathbf{h}_t})\right)}{\sum_{\Theta_{\K}}\mu_1^{\K}(\overline{\ell}_{-\K})\exp\left(-t\Delta K(\overline{\ell}_{-\K}, \bm{\nu}_{\mathbf{h}_t}) \right) }.
    \label{eq:epsilon-eta}
\end{align}

Consider the numerator in \eqref{eq:epsilon-eta}. Note that $x e^{-tx}$ is decreasing for all $x> t^{-1}$ and $\Delta K(\widehat{\ell}_{-\K}, \bm{\nu}_{\mathbf{h}_t})\geq \epsilon$ for all $\widehat{\ell}\in \Theta_{\K}^\epsilon$. Hence, for any $t\geq \max[t_{\eta}, \epsilon^{-1}]$, the numerator in \eqref{eq:epsilon-eta} is upper-bounded by $\epsilon e^{-t\epsilon}$. As for the denominator, we claim that it is uniformly lower bound by a positive constant. To see this, we first recall that $\Theta_{\K}^\star(\bm{\nu}_{\mathbf{h}_t})$ is always non-empty since $\Theta_{\K}$ is finite. Thus, we have
\begin{align*}
    &\sum_{\Theta_{\K}}\mu_1^{\K}(\overline{\ell}_{-\K})\exp\left(-t\Delta K(\overline{\ell}_{-\K}, \bm{\nu}_{\mathbf{h}_t}) \right)  \\
    \geq & \sum_{\Theta_{\K}^\star(\bm{\nu}_{\mathbf{h}_t})}\mu_1^{\K}(\overline{\ell}_{-\K})\exp\left(-t\Delta K(\overline{\ell}_{-\K}, \bm{\nu}_{\mathbf{h}_t}) \right)\\
    =& \sum_{\Theta_{\K}^\star(\bm{\nu}_{\mathbf{h}_t})}\mu_1^{\K}(\overline{\ell}_{-\K})> \min_{\overline{\ell}_{-\K}\in \Theta_{\K}^\star(\bm{\nu}_{\mathbf{h}_t})} \mu_{1}^\K(\overline{\ell}_{-\K})\geq \min_{\overline{\ell}_{-\K}\in \Theta_{\K}} \mu_{1}^\K(\overline{\ell}_{-\K}),
\end{align*}
where the equality follows $\Delta K(\overline{\ell}_{-\K}, \bm{\nu}_{\mathbf{h}_t})=0$, for $\overline{\ell}_{-\K}\in \Theta_{\K}^\star$. Since $\mu_1^\K$ has full support by \Cref{assumption:bayes},  $\min_{\overline{\ell}_{-\K}\in \Theta_{\K}} \mu_{1}^\K(\overline{\ell}_{-\K})$ is a strictly positive constant, denoted by $k$. Then,  $\mathrm{(\ast)} \leq e^{2t\eta}\epsilon e^{-t\epsilon}k^{-1}$. Let $\eta=\frac{\epsilon}{4}$. Then $e^{2t\eta}\epsilon e^{-t\epsilon}k^{-1}=e^{\frac{-t\epsilon}{2}}\epsilon k^{-1}$, which converges to $0$ as $t \rightarrow \infty$. This completes the proof of \eqref{eq:upper-bound convergence}, and hence, \eqref{eq:asym-consistency} holds.
\end{proof}

\subsection{Full Proof of \Cref{coro:convergence-berk-nash}}
\begin{proof}
Due to the i.i.d. states, the belief state $\mathbf{b}_t^{\K}$ is also i.i.d., and hence, the empirical occupancy measure converges almost surely by the Glivenko-Cantelli theorem. Denote by $\bm{\nu}$ the limit point of the joint belief occupancy. According to \Cref{thm:asym-consistency}, $\mu_t^{\K}$ asymptotically concentrates on $\Theta_{\K}^\star(\bm{\nu})$, which, however,  does not necessarily imply convergence of $(\mu_t^{\K})_{t\geq 1}$. Thanks to the sequential compactness of $\Delta(\Theta_{\K}^\star)$, there exists a convergent subsequence of $(\mu_t^{\K})_{t\geq 1}$ whose limit point is denoted by $\mu_{\infty}^{\K}\in \Delta(\Theta_{\K}^\star)$. The rest is to show that the limit point of the empirical strategy profile, which is denoted by $\pi_{\K}$, satisfies (i). Suppose by contradiction that $\pi_{\K}(\cdot|o^{\K})\notin \argmin\E_{\mathbb{P}[\cdot\mid \mu^{\K}_{\infty}, \mathbf{b}^{\K}]}[c^{\K}(S, A^{\K}, A^{-\K})]$ for some $o^{\K}$. Then, there exists a $T>0$ such that there is at least one subsequence of $(\pi_{\K,t})_{t\geq T}$ by Alg.~\ref{alg:online_rollout} (denoted also by $\pi_{\K,t}$), satisfying $\pi_{\K,t}(\cdot|o^{\K})\notin \argmin\E_{\mathbb{P}[\cdot\mid \mu^{\K}_{\infty}, \mathbf{b}^{\K}]}[c^{\K}(S, A^{\K}, A^{-\K})]$. Let $g(\pi_{\K},\mu)\triangleq \E_{\mathbb{P}[\cdot\mid \mu, \mathbf{b}^{\K}]}[c^{\K}(S, A^{\K}, A^{-\K})]$, a bilinear function defined over a compact set $\Delta(\mathcal{\mathcal{A}^{\K}})\times \Delta(\Theta_{\K})$, and we write $g_{\mu}(\pi_{\K})$ when fixing the second variable. Due to its bilinearity, we obtain that 1) the level set of $g_{\mu_t}$ is bounded (the entire domain is bounded) and 2) $g_{\mu_t}$ pointwise converges to $g_{\mu_{\infty}}$, i.e., for any $\pi_{t}\rightarrow \pi$, $g_{\mu_t}(\pi_t)\rightarrow g_{\mu_{\infty}}(\pi)$. Hence, the epigraph of $g_{\mu_t}$ converges to $g_{\mu_{\infty}}$'s in the Painlev\'e-Kuratowski sense\cite[Thm 7.11]{RockWets98}. Invoking \cite[Thm 7.33]{RockWets98}, we have  $\limsup_{t} \argmin g_{\mu_t}\subset \argmin g_{\mu_{\infty}}$. However, $\pi_{\K, t}$ cannot be simultaneously included in $\argmin g_{\mu_t}$ by the conjectural best response while excluded by the left-hand side, which is derived from the reductio hypothesis. Therefore, $\pi_{\K}$ satisfies the optimality condition.
\end{proof}
\subsection{Experiment Setup}
\label{app:IT-config}
\setcounter{table}{0}
\renewcommand{\thetable}{\Alph{subsection}.\arabic{table}}
The IT infrastructure configuration includes 64 servers and an Intrusion Detection System (\textsc{ids}) that logs events in real-time \cite{tifs_hlsz_supplementary}. Some of the serves are vulnerable to \textsc{apt}s, and we emulate the \textsc{apt} attacker using the actions listed in \Cref{tab:attacker_actions}. The complete infrastructure configuration can be found at \cite[Appendix G]{tifs_hlsz_supplementary}. Clients access the services through a public gateway, which is also open to the attacker. The defender monitors the infrastructure by observing \textsc{ids} alerts. The partial observation $o_t$ corresponds to the priority-weighted sum of the number of \textsc{ids} alerts at time $t$. The partial observation kernel $z$ is estimated through the empirical distribution using $10^5$ measurements of $o_t$ from a digital twin of the IT infrastructure. The measurements are available at \cite{csle_source_code}, where $|\mathcal{O}|=26178$. We remark that such a large observation space leads to a prolonged offline computation for state-of-the-art methods, such as heuristic search value iteration \cite{horak23posg} and proximal policy optimization (\textsc{ppo}) \cite{ppo}. In stark contrast, the proposed \textsc{col} can proceed in real time.
\begin{table}[!h]
\centering
\begin{tabular}{ll} \toprule
  {\textit{Type}} & {\textit{Actions}} \\ \midrule
  Reconnaissance  & \textsc{tcpp} \textsc{syn} scan, \textsc{udp} port scan, \textsc{tcpp} \textsc{xmas} scan\\
                  & \textsc{vulscan} vulnerability scanner, ping-scan \\
  &\\
  Brute-force & \textsc{telnet}, \textsc{ssh}, \textsc{ftp}, \textsc{cassandra},\\
                  &  \textsc{irc}, \textsc{mongo}, \textsc{mysql}, \textsc{smtp}, \textsc{postgres}\\
                  &\\
  Exploit & \textsc{cve}-2017-7494, \textsc{cve}-2015-3306,\\
                  & \textsc{cve}-2010-0426, \textsc{cve}-2015-5602,\textsc{cve}-2015-1427 \\
                  &  \textsc{cve}-2014-6271, \textsc{cve}-2016-10033\\
                  & \textsc{cwe}-89 weakness on the web app \textsc{dvwa}\\
  \bottomrule\\
\end{tabular}
\caption{Attacker actions in simulation; when the attacker takes the stop action ($\mathsf{S}$), a randomly selected action is applied to every reachable server; actions are identified by the vulnerability identifiers in the Common Vulnerabilities and Exposures (\textsc{cve}) database and the Common Weakness Enumeration (\textsc{cwe}) list available at \url{https://cve.mitre.org/}; the web app \textsc{dvwa} is available at \url{https://github.com/digininja/DVWA}.}\label{tab:attacker_actions}
\end{table}

The hyperparameters of \textsc{ppo} baselines follow \cite[Table 5]{tifs_hlsz_supplementary} and are listed in \Cref{tab:ppo-hyper}, which are obtained from grid search.
\begin{table}[!h]
\centering
\begin{tabular}{ll}
\toprule
   learning rate  & $10^{-5}$ \\
   Adam stepsize  & $4\times 10^{-3}$\\
   batch size     & 64\\
   \# layers      & 4 \\
   \# neurons     & 64 \\
   clip           & $0.2$\\
   GAE            & $0.95$\\
   entropy coeff  & $10^{-4}$\\
   activation     & ReLU\\
\bottomrule
\end{tabular}
\caption{The hyperparameters of \textsc{ppo} baselines in Fig.~\ref{fig:evaluation_results}.}
\label{tab:ppo-hyper}
\end{table}

\end{appendix}

\end{document}